\documentclass[aps,superscriptaddress,twocolumn,floatfix]{revtex4-1}  

\usepackage{hyperref}
\usepackage{enumitem}
\usepackage{amsmath,amsfonts,amssymb,amsthm,graphics,graphicx,epsfig,times,bbm}
\usepackage[dvipsnames]{xcolor}
\usepackage{jrnbsymb}
\usepackage[sort&compress]{natbib}

\newcommand\sectionref[2][Section]{\hyperref[#2]{#1~\ref*{#2}}}
\renewcommand\autopageref[1]{\hyperref[#1]{page~\pageref*{#1}}}

\newcommand\ket[1]		{\left\lvert#1\right\rangle\mspace{-1.5mu}}		
\newcommand\bra[1]		{\mspace{-1.5mu}\left\langle#1\right\rvert}		
\newcommand\bracket[2]	{\left\langle#1\mspace{-1mu}\left|\mspace{1mu}#2\right\rangle\right.}	
\newcommand\card[1]		{\left\lvert #1 \right\rvert}							

\newcommand\ens[1]		{\left\{#1\right\}}										
\newcommand\paren[1]		{\left( #1 \right)}		     						 	
\newcommand\sqparen[1]	{\left[ #1 \right]}		   					  		
\newcommand\anglebr[1]	{\left\langle #1 \right\rangle}

\newcommand\sox[1]{^{\ox #1}}

\let\H\relax
\namecommand[c]{\H}
\newcommand\cH[1][2]{\H_{#1}}

\let\L\relax
\namecommands[sf]{\L, \U}
\namecommand\cU

\newcommand{\rr}{\mathbbm{R}}

\newcommand\ssfStyle[1]{{\scalefont{0.9}\textmd{\textup{\textsf{#1}}}}}
\NewNameCommandStyle{ComplexityClass}{\ssfStyle{#1}}
\NewNameCommandStyle{MatrixGroup}{\mathrm{\uppercase{#1}}}

\let\P\relax
\namecommands[ComplexityClass]{\P, \NP,\QMA}
\namecommands[MatrixGroup]{\GL}
\NameMathOperators{\Symm, \rank, \Exp*}

\def\unqm2sat{\textsc{unique quantum 2-sat}}
\def\qm#1sat{\textsc{quantum ${#1}$-sat}}

\renewcommand{\vec}[1]{\text{\boldmath$#1$}}

\newtheorem{lemma}{Lemma}

\newtheorem{proposition}{Proposition}

\renewenvironment{cases}[1][c]{%
 	\left\{ \begin{array}{#1@{\quad}l}%
}{%
	\end{array} \right.%
}

\setitemize{listparindent=\parindent}

\newcounter{descenum}
\makeatletter
\newcounter{inlinum}
\renewcommand\theinlinum{\textup{(\alph{inlinum})}}

	\def\@after@inlinum{%
		\catcode`\^^M=10%
		\gdef\@tempa{\catcode`\^^M=5}%
		\expandafter\@tempa\ignorespaces}%
\makeatother

\usepackage{color}
\def\id{\mathbbm{1}}

\newtheorem*{lemma*}{Proposition}

\makeatletter
\renewcommand\eqref[1]{\hyperref[#1]{Eq.~\textup{\tagform@{\ref*{#1}}}}}
\makeatother

\begin{document}

\title{Ground states of unfrustrated spin Hamiltonians satisfy an area law}

\author{Niel de Beaudrap}
\affiliation{Institute of Physics and Astronomy, University of Potsdam, 14476 Potsdam, Germany}

\author{Tobias J.\ Osborne}

\affiliation{Institute for Advanced Study Berlin, 14193 Berlin, Germany}

\author{Jens Eisert} 

\affiliation{Institute of Physics and Astronomy, University of Potsdam, 14476 Potsdam, Germany}
\affiliation{Institute for Advanced Study Berlin, 14193 Berlin, Germany}

\date{\today}

\vspace*{-.6cm}

\begin{abstract}
	We show that ground states of unfrustrated quantum spin-$1\!/2$ systems on general lattices satisfy an entanglement 
	area law, provided that the Hamiltonian can be decomposed into nearest-neighbor interaction terms which have entangled excited states.
	The ground state manifold can be efficiently described as the image of a low-dimensional subspace of low Schmidt measure, 
	under an efficiently contractible tree-tensor network.
	This structure gives rise to the possibility of efficiently simulating the complete ground space (which is in general degenerate).
	We briefly discuss ``non-generic'' cases, including highly degenerate interactions with product eigenbases, using a relationship to percolation theory.
	We finally assess the possibility of using such tree tensor networks to simulate almost frustration-free spin models.
\end{abstract}
\maketitle

\section{Introduction}

An important insight in the study of quantum many-body systems is related to
the observation that common states that naturally occur do not quite exhaust the entire
Hilbert space available to them, but instead a much smaller subspace. This 
insight is at the heart of powerful numerical methods that have 
been devised in recent years. Ideas such as the density-matrix renormalization group
approach, and new ideas that allow for the simulation of higher-dimensional quantum
lattice models \cite{Review,Scholl,2d}, work exactly because they model well quantum 
states that in a certain sense have little entanglement.
More precisely, the states which are tractible by these approaches satisfy what is called
an \emph{area law} \cite{Review,Wilczek,Bombelli,Srednicki,Harmonic,Latorre,Area1,Area2,Jin,Cardy,Fermi1,Fermi2,Fermi3,Hastings,Masanes}, 
so the entropy of a subregion scales at most as the boundary area of
that region (for a review, see Ref.~\cite{Review}). 
For practical purposes, and in particular for 1D systems, these methods
in particular give accurate accounts of ground state properties. 

Now, not all ground states of local quantum lattice models can be efficiently approximated. 
This holds even true for 1D chains: indeed, one can construct models for which approximating 
the ground state energy
is provably \NP-hard \cite{NPH} --- albeit using a fairly sophisticated construction involinvg large local dimensions~\cite{Power}. 
An important feature of these constructions is that the difficulty of their 
solution appears to be strongly related to whether the system is \emph{frustrated} or not.
This suggests that whether or not the system is frustrated is another criterion for whether a 
quantum lattice model should be considered ``easy''  or ``hard'', 
in addition to its ground states having ``a lot'' or ``little'' entanglement.
This intuition that frustrated  systems should be hard to simulate is indeed true for
classical systems, where the frustrated or glassy models are the hard ones to 
describe. For quantum systems, there is evidence that the situation should be more complex
\cite{Shor}.

In this work, we explore a class of models where the intuition of frustration-free models being 
\emph{easy to solve} holds true.
Building upon work in Ref.~\cite{Bravyi06} and in Ref.\ \cite{OurPRL}, 
for a natural class of two-local Hamiltonians acting on spin-$1/2$ particles (simply ``spins'' henceforth), we show that ground states can be reduced to a completely characterized and low-dimensional subspace, and then re-constructed by identifying the ground state-space of each interaction of the Hamiltonian term-by-term.
Specifically, the ground space is the image of a symmetric subspace under an explicitly constructible, and efficiently contractible, tensor network.
It follows that the ground states satisfy an area law, and hence contain little entanglement in the above sense.
This generalizes recent results regarding the \emph{existence} of states which have little entanglement in the 
ground-state manifolds of such Hamiltonians~\cite{NoGoMBQC}.
We discuss how to efficiently simulate the ground state manifold, and suggest how this 
could be used to simulate ``almost'' frustration-free quantum lattice models.

\section{Preliminaries}
\subsection{Frustration-free Hamiltonians and area laws}
\label{sec:prelimsFrustrationFree}

We consider \emph{spin-$1/2$ Hamiltonians on a lattice}. The lattice is described by some graph,
the vertex set of which we denote by $V$. Naturally, the Hamiltonian will be local, or more specifically
include only nearest-neighbor interaction terms. We represent the Hamiltonian as
		\begin{equation}
		 		H \;=\! \sum_{\ens{a,b} \subset V}\! h_{a,b}
		\end{equation}
for some terms $h_{a,b}$ acting on pairs of spins in the lattice described by $V$.
By rescaling, we may without loss of generality require that the ground state energy of each interaction term $h_{a,b}$ is zero.
We wish to describe properties of the ground state manifold $M$ of such Hamiltonians, given the list of the individual two-spin terms $h_{a,b}$ as input.
An important class of Hamiltonians are those which are \emph{frustration-free} (or \emph{unfrustrated}), for which each ground state vector $\ket{\Phi} \in M$ is also a ground state of the individual coupling terms: i{.}e{.} for which
\begin{align}
			h_{a,b}\ket{\Phi} =0
\end{align}
holds for all  $h_{a,b}$ and all $|\Phi\rangle \in M$. 
The actual ground state $\rho$ is the maximally mixed state over $M$, and so is mixed unless the ground state manifold $M$ is non-degenerate.

Our main results pertain to frustration-free spin Hamiltonians as above, with the further constraint that each term $h_{a,b}$ has at least one entangled excited state.
In \sectionref{sec:areaLaws}, we show that the ground states $\rho$ of such Hamiltonians satisfy an area law \cite{Review,Wilczek,Bombelli,Srednicki,Harmonic,Latorre,Area1,Area2,Jin,Cardy,Fermi1,Fermi2,Fermi3,Hastings,Masanes}: that is, for a contiguous region of spins $A \subset V$, the {\it entanglement of formation}.
	If the ground state is non-degenerate and hence pure, the  entanglement of formation is nothing but the usual entanglement entropy.
of the ground state satisfies
\begin{align}
			E_F(\rho) \le C |\partial A|
\end{align} 
for $C > 0$ constant, where $|\partial A|$ is the ``boundary area'' (i.{}e.{} the number of edges in the interaction graph of $H$, which are incident to both $A$ and $V \setminus A$): see \autoref{Area}.
The {\it entanglement of formation} is the largest asymptotically continuous entanglement monotone, so this also implies an entanglement area law for e.{}g.{} the {\it distillable entanglement}.
Thus, ground states of frustration-free Hamiltonians contain little entanglement.
\begin{figure}[hbt]
\includegraphics[width=.5\columnwidth]{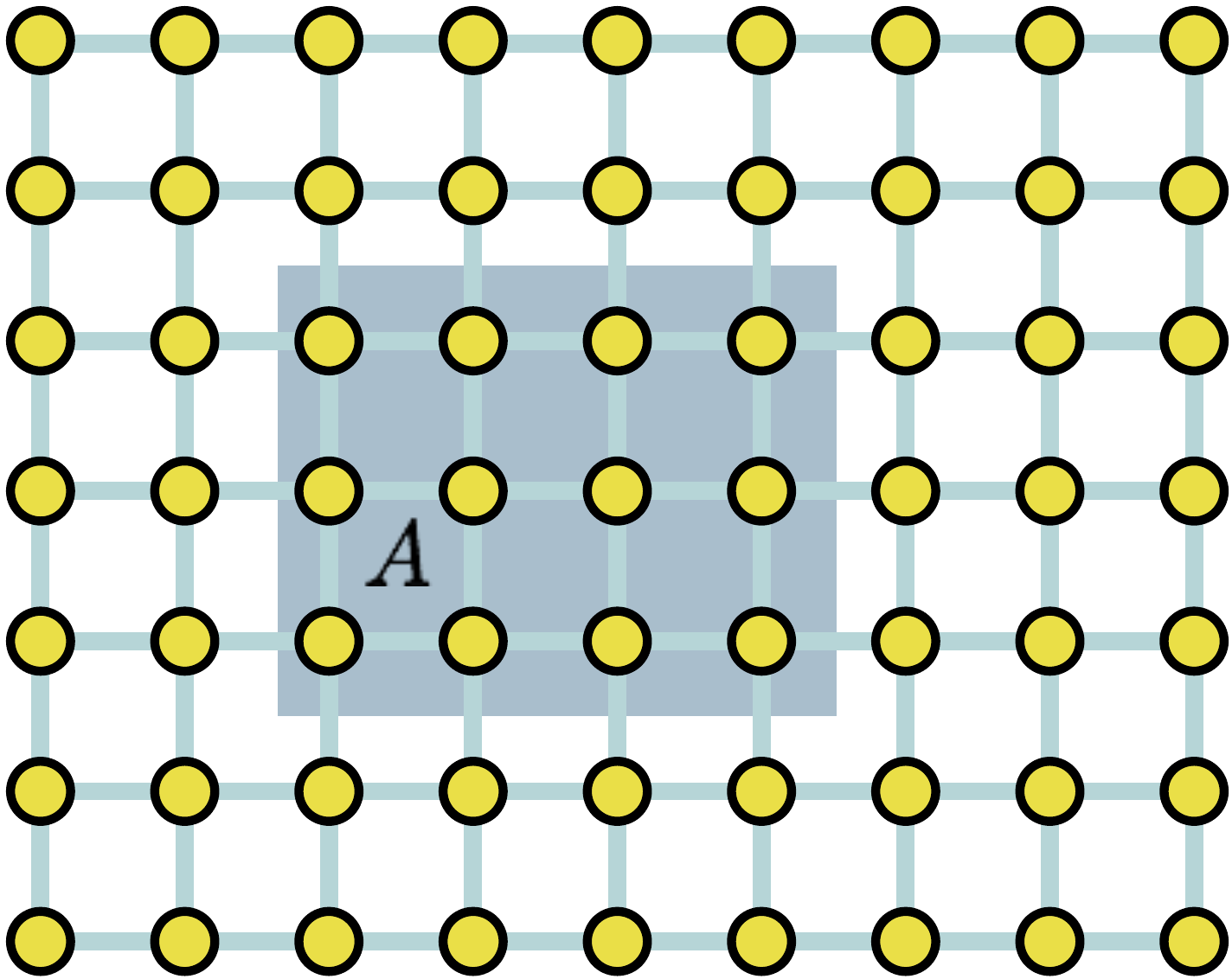}
\caption{The entanglement scales at most as the boundary area $|\partial A| $
of a distinguished region $A$ in some lattice, here a cubic lattice, with tighter bounds for special
cases.\label{Area}}
\end{figure}

This follows as a consequence of the fact that the ground-state manifold is the image of the symmetric subspace on $n$ spins (for some $n$ bounded above by the number of spins of the Hamiltonian) under an 
efficiently simulatable, and explicitly constructible, tree-tensor network.
We show this in \sectionref{sec:sampling}, in which we demonstrate how this allows us to efficiently simulate the ground space of the spin models we consider. For the cases where we have a tree tensor network with a single top root, the problem
considered here may be viewed as exactly the converse problem to the one discussed in Ref.\ \cite{Giov}.

Our result of an analytical area law complements results on area laws in harmonic bosonic systems \cite{Harmonic,Area1,Area2}, 	fermionic \cite{Fermi1,Fermi2,Fermi3} on cubic lattice and general gapped models in one-dimensional quantum chains~\cite{Hastings}.
For a comprehensive review on area laws --- and on implications on the simulatability of quantum 
many-body systems --- see Ref.~\cite{Review}.

\subsection{Quantum \textsc{2-sat} problem}

The arguments behind our analysis builds upon and extends the ideas of Ref.~\cite{Bravyi06}, which defined the problem of quantum satisfiability, and presented Bravyi's algorithm for \qm2sat.
We describe here the connection between this problem and frustration in spin Hamiltonians.

\qm2sat\ is the quantum analogue of the ``classical'' 2-\textsc{sat} problem on boolean formulae.
The latter asks when there exists an assignment of boolean variables $x_1, \ldots, x_n$ which simultaneously ``satisfy'' 
a collection of constraints on pairs of those variables.
This problem is efficiently solvable~\cite{APT79}; in contrast, the 
similar problem 3-\textsc{sat} (in which constraints apply to triples of variables) is \NP-complete~\cite{Karp72}.
In \qm2sat, individual clauses on boolean variables are replaced by projectors $p_{a,b}$ with
\begin{align}
	p^2_{a,b} = p_{a,b}
\end{align}
on pairs of spins: an instance of \qm2sat\ is \emph{satisfiable} if there is a vector which is a zero eigenvector of each projector simultaneously.

For an instance of \qm2sat, determining whether there exist such simultaneous zero eigenvectors is equivalent to determining whether the Hamiltonian obtained by summing the projectors is unfrustrated.
Conversely, the problem of determining when a $2$-local spin Hamiltonian $H$ is frustration-free may be reduced to \qm2sat, by rescaling the terms of the Hamiltonian $H$ so that each term $h_{a,b}$ has a minimal eigenvalue of $0$, and replacing each rescaled term $h_{a,b}$ with the projector $p_{a,b}$ onto $\img(h_{a,b})$. 
By construction, such a substituation does not affect the ground space of the terms.
Thus, solving \qm2sat\ is equivalent to determining whether a $2$-local spin Hamiltonian is frustration-free.

Recently, random instances of $\qm ksat$ with rank-$1$ projectors \cite{Footnote}
have been studied for $k \ge 2$, delineating the ``boundary'' of 
frustration in $k$-local spin Hamiltonians in terms of the density of 
interactions~\cite{kQSat,RandomQSat,BravyiQSat,Zittartz}.
We instead extend the findings of Ref.~\cite{Bravyi06} for $k = 2$, remarking on implications for simulating the ground space manifold in frustration-free Hamiltonians.
In \sectionref{sec:reviewBravyi}, we review Bravyi's algorithm for \qm2sat, in order to demonstrate important features of the reductions involved when they are applied to unfrustrated Hamiltonians satisfying natural constraints. 

\section{Reduction tools for frustration-free Hamiltonians}
\label{sec:reviewBravyi}

Bravyi's algorithm for \qm2sat~\cite{Bravyi06} efficiently demonstrates the satisfiability of an instance of \qm2sat\ by a sequence of reductions of Hamiltonians, yielding a \emph{homogeneous} instance (in which all projectors have rank $1$), and then verifying the satisfiability of these instances.
We may similarly use Bravyi's algorithm to detect frustration in $2$-local spin Hamiltonians, and consider the features of these Hamiltonian reductions when applied to particular classes of frustration-free Hamiltonians.

Throughout the following, we admit representations of Hamiltonians $H'$ with non-zero single spin terms $h_a$,
\begin{equation}
		H'
	\;=\!	\sum_{\ens{a,b} \subset V} \!\! h_{a,b} \;\;+\;\; \sum_{a \in V} h_a	\;,
\end{equation}
and again describe $H'$ as unfrustrated if there exists a joint ground state with eigenvalue zero of all terms 
(including the single-spin terms $h_a$).

\subsection{Reductions by isometries}
\label{sec:reductionsByIsometries}

Condensing the analysis of Ref.~\cite{Bravyi06}, we consider a reduction for $2$-local Hamiltonians $H$ to Hamiltonians $H'$ on fewer spins, provided that $H$ contains only positive semidefinite terms which have non-trivial kernels. Throughout, we denote $\C^2$ by $\cH$.

\subsubsection{Two-spin isometric contractions}
\label{sec:twoQubitIsometries}

Consider a Hamiltonian term $h_{u,v }$ of rank $2$ or $3$. 
If $H$ is frustration-free, $h_{u,v}$ fixes a subspace of $\cH\sox{\ens{u,v}}$ of dimension at most $2$, over which the reduced state $\rho_{u,v}$ of a state vector $\ket{\Phi} \in \ker(H)$ must be a mixture.
We describe this reduced state by an encoding of one spin into two.
Let $\ens{\ket{\psi_0}, \ket{\psi_1}, \ket{\psi_2}, \ket{\psi_3}}$ be an orthonormal basis for $\cH \ox \cH$ such that
\begin{equation}
		\ker(h_{u,v}) \subset \Span \ens{\ket{\psi_0}, \ket{\psi_1}}.
\end{equation}	
Define an isometry $U_{u:u,v}: \cH\sox{\ens{u}} \!\to\! \cH\sox{\ens{u,v}}$ such that
\begin{align}
	\label{eqn:twoQubitIsometry}
 		\sum_{x \in \ens{0,1}} \!\! \alpha_x \ket{x}_u
	\quad\mapsto&\;\;
		\sum_{x \in \ens{0,1}} \!\! \alpha_x \ket{\psi_x}_{u,v} 	\;.
\end{align}
This is an \emph{isometric reduction}, similar to those in a \emph{tree tensor network} or a 
\emph{MERA ansatz}~\cite{MERA}.
By construction, the support of the reduced state $\rho_{u,v}$ lies in $\img(U_{u:uv})$. 
We may then define a Hamiltonian 
\begin{align}
		\label{eqn:twoQubitIsometryInducedHamiltonian}
		H'  =  U_{u:u,v}\herm  H  U_{u:u,v} 
\end{align}
on the subsystem $V' = V \setminus \ens{v}$, where the spin $v$ is essentially deleted;  any state vector $\ket{\Phi} \in \ker(H)$ then has the form 
\begin{align}
		\ket{\Phi} = U_{u:uv} \ket{\Phi'}, \qquad\ket{\Phi'} \in \ker(H').
\end{align}
We may express $H'$ as a sum of terms 
\begin{equation}
	h'_{a,b} = U_{u:u,v}\herm h_{a,b}  U_{u:u,v},\quad
	h'_a = U\herm_{u:u,v}h'_aU_{u:u,v}.
\end{equation}	
(In the case that $h$ is of rank $3$, these will include a non-zero single spin operator $h'_{u,v}$ which acts on $u$ alone.)
If $H$ contains non-zero terms $h_{a,u}$ and $h_{a,v}$, we obtain two terms $h'_{a,u} = U\herm_{u:u,v}h_{a,u}U_{u:u,v}$ and $h'_{a,v} = U\herm_{u:u,v}  h_{a,v}  U_{u:u,v}$ in the Hamiltonian $H$, which both act on the spins $u$ and $a$.
We sum these to obtain a combined term
\begin{equation}
		\bar h'_{a,u} = h'_{a,u} + h'_{a,v} 
\end{equation}
in the reduced Hamiltonian, which may be of higher rank than either $h'_{a,u}$ or $h'_{a,v}$.
(We similarly accumulate any single-spin contributions $h'_u$ and $h'_v$ which may arise from single-spin terms $h_u$ and $h_v$ in $H$.) 
Fig.\ \ref{fig:graphReduction} illustrates the
effect of multiple reductions on the interaction graph of the
Hamiltonian.

\begin{figure}
	\begin{center}
	\includegraphics[width=.42\columnwidth]{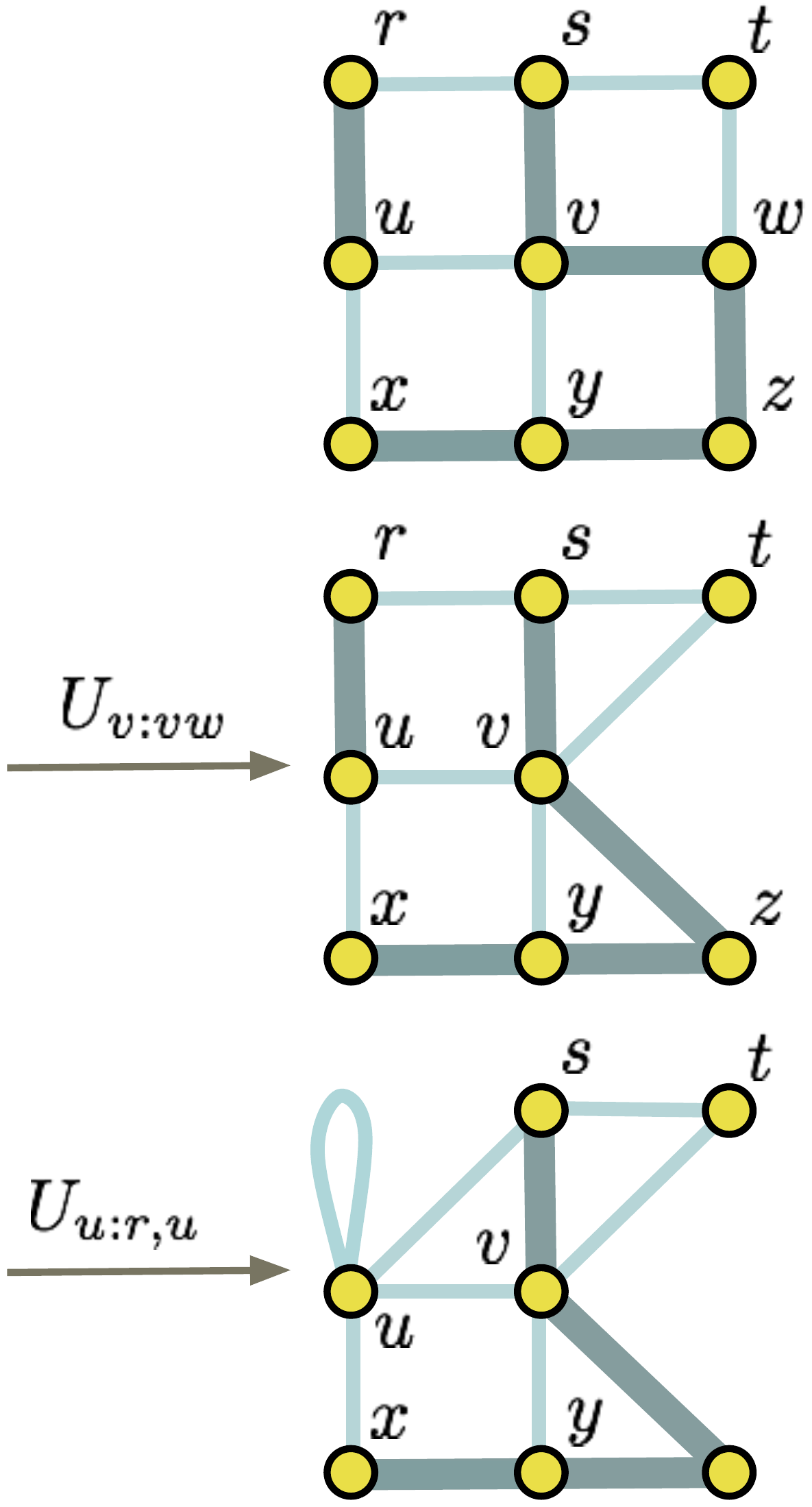}
	\caption{
		Illustration of transformations of the interaction graph of a Hamiltonian $H$ by two-spin isometries.
		Darker, thicker edges represent interaction terms $h_{a,b}$ of rank $2$ or $3$, which may be eliminated by contraction (corresponding to a reduction $H \mapsto* U_{a:ab}\herm H U_{a:ab}$).
		In the case of an interaction term with $\rank(h_{r,u}) = 3$ as illustrated above, a single spin operator (represented above by a loop) is produced on the contracted vertex.}\label{fig:graphReduction}
	\end{center}
\end{figure}

For example: Consider an anti-ferromagnetic four-spin Hamiltonian, with interactions 
\begin{equation}
	h_{j,k} \;=\; \tfrac{1}{2}\sigma_x^{(j)} \sigma_x^{(k)} \;+\; \tfrac{1}{2}\sigma_y^{(j)} \sigma_y^{(k)} 
\end{equation}	
on a four-spin cycle with interacting pairs $(1,2)$, $(2,3)$, $(3,4)$, and $(4,1)$.
Here we denote
\begin{subequations}
\begin{align}
	\sigma_x \;=&\; \left[\;\begin{matrix} 0 & 1 \\ 1 & 0 \end{matrix}\;\right] \;=\; \ket{1}\bra{0} + \ket{0}\bra{1}	,		\\
	\sigma_y \;=&\; \left[\;\begin{matrix} 0 & \!-i \\ i & 0 \end{matrix}\;\right] \;=\; i\ket{1}\bra{0} - i\ket{0}\bra{1}	.
\end{align}
\end{subequations}
Rescaling the interactions to have ground energy zero (and taking these for the $h_{j,k}$ instead) gives us
\begin{align}
		h_{j,k}	\;=\;	\ket{0,0}\bra{0,0}_{j,k} + \ket{1,1}\bra{1,1}_{j,k}	\;.
\end{align}
The kernel of this operator is clearly spanned by $\ket{0,1}$ and $\ket{1,0}$.
We may consider the effect of contracting spins $3$ and $4$ into a renormalized spin, using the isometry
\begin{align}
	R_{3:3,4} = \big(\ket{0}_3 \ox \ket{1}_4\big) \bra{0}_3 + \big(\ket{1}_3 \ox \ket{0}_4\big) \bra{1}_3 \;.
\end{align}
By construction, we have $R_{3:3,4}\herm h_{3,4} R_{3:3,4} = 0$; and as $\ens{1,2}$ is disjoint from $\ens{3,4}$, we have $h'_{1,2} = R_{3:3,4}\herm h_{1,2} R_{3:3,4} = h_{1,2}$.
We compute the renormalized terms $h'_{2,3}$ and $h'_{4,1}$ as
\begin{subequations}
\begin{align}
		h'_{2,3}
	\;=&\;
			R_{3:3,4}\herm \Big( \ket{0,0}\bra{0,0}_{2,3} \ox \idop_4 + \ket{1,1}\bra{1,1}_{2,3} \ox \idop_4 \Big) R_{3:3,4}
	\notag\\=&\;
			\ket{0,0}\bra{0,0}_{2,3} + \ket{1,1}\bra{1,1}_{2,3}	\;,
	\\[1ex]
		h'_{4,1}
	\;=&\;
			R_{3:3,4}\herm \Big( \idop_3 \ox \ket{0,0}\bra{0,0}_{4,1}  +  \idop_3 \ox \ket{1,1}\bra{1,1}_{4,1} \Big) R_{3:3,4}
	\notag\\=&\;
			\ket{1,0}\bra{1,0}_{3,1} + \ket{0,1}\bra{0,1}_{3,1}	\;;
\end{align}
\end{subequations}
up to rescaling, the resulting renormalized Hamiltonian is then
\begin{align}
		H'
	=
		R_{3:3,4}\herm H R_{3:3,4}
	\sim&\;
				\tfrac{1}{2}\Big(	\big[ \sigma_x^{(1)} \sigma_x^{(2)} + \sigma_y^{(1)} \sigma_y^{(2)}	 \big]	\notag\\
				&\;\;+\;		
				\big[ \sigma_x^{(2)} \sigma_x^{(3)} + \sigma_y^{(2)} \sigma_y^{(3)}	\big]		\notag\\
				&\;\;-\;
				\big[ \sigma_x^{(3)} \sigma_x^{(1)} + \sigma_y^{(3)} \sigma_y^{(1)}	\big]\Big)\;,
\end{align}
with a ferromagnetic coupling between the site $1$ and the renormalized site $3$.
Arbitrary rank-2 or rank-3 interactions may be contracted similarly.
 
\subsubsection{Single-spin deletions}
\label{sec:oneQubitRemoval}

	For a $2$-local Hamiltonian $H$ containing non-zero single-spin operators $h_v$ (e.g., such as may arise from the preceding reduction),
	any state vector $\ket{\Phi} \in \ker(H)$ must be factorizable into a single-spin pure state vector 
	$\ket{\psi} \in \ker(h_v)$ acting on $v$, and some state of the remaining spins.
	If the operator $h_v$ has full rank, it follows $\ket{\psi} = \vec 0$, so that $H$ has trivial kernel.
	Otherwise, $\ket{\psi}$ may be taken to be a unit vector spanning the kernel of $h_v$, and we may form a Hamiltonian
	\begin{align}
		\label{eqn:singleQubitRemoval}
		 H' = \bigl( \bra{\psi}_v \ox  \id \bigr)  H  \bigl( \ket{\psi}_v \ox  \id \bigr)
	\end{align}
	on the subsystem $V' = V \setminus \ens{v}$, consisting of a sum of terms $h'_a = h_a$ acting on individual spins $a \ne v$, and terms
	\begin{gather}
	 		h'_{a,b}
		 = 
			\bigl( \bra{\psi}_{v} \ox  \id \bigr)  h_{a,b}  \bigl( \ket{\psi}_{v} \ox  \id \bigr),
	\end{gather}
	acting on pairs of spins $\ens{a,b} \subset V$.
	In the latter case, if $v \in \ens{a,b}$, then $h'_{a,b}$ will be an operator acting on a single spin; otherwise, we have $h'_{a,b} = h_{a,b}$.
	Again, in the case of single-spin terms $h'_{a,b}$ acting on a spin $a$, if there was a term $h_a$ present in $H$, we may accumulate these into a term $\bar h'_a = h_a + h'_{a,b}$.

We may again describe $H'$ using an isometric reduction in this case: if we define $P_v = \ket{\psi_v} \ox \id_{V'}$, then $P_v$ is an isometry whose image contains any state vector $\ket{\Phi} \in \ker(H)$.
We may then rewrite	 \eqref{eqn:singleQubitRemoval} as
\begin{align}
	H' = P_v\herm H P_v.
\end{align}
expressing it in a form more explicitly similar to \eqref{eqn:twoQubitIsometryInducedHamiltonian}.
	
\subsubsection{Remarks on these reductions}

The reductions above correspond to the reductions presented in Ref.~\cite{Bravyi06} for instances of \qm2sat\ which contain two-spin operators of 
rank greater than $1$.
This allows us to reduce \qm2sat\ to the special case of ``homogeneous'' instances (in which all terms have rank $1$).

The key feature of both reductions above is that the kernels of $H$ and $H'$ are related by isometries, and so have the same dimension.
If the Hamiltonian $H'$ has any terms of full rank (acting on either one or two spins), it follows that the Hamiltonian $H'$ has trivial kernel; then the same holds for $H$ as well.
If we do not encounter any full-rank terms, each reduc- tion produces a Hamiltonian acting on one fewer spins, even- tually yielding a ``homogeneous'' Hamiltonian (extending the terminology of Ref.~\cite{Bravyi06} to Hamiltonians in general, including those with single-spin terms of rank 1).

The choice of the reduction at each stage does not matter, in the following sense.
As long as we have a Hamiltonian containing two-spin terms of rank at least $2$, and which does not contain full-rank terms, we may extend any sequence of reductions to one which terminates with a Hamiltonian $\tilde H$ which is either homogeneous or contains a full-rank term.
In the latter case, the original Hamiltonian has trivial kernel, and is therefore frustrated; otherwise, we obtain a ``homogeneous'' Hamiltonian whose kernel may be mapped to that of the original Hamiltonain $H$ by a sequence of known isometries.
If we can solve the homogeneous case, we may then choose the reductions according to whichever criteria are convenient.

\label{remark:isometriesTreeTensorNetwork}
Note that this reduction process, from an input Hamiltonian $H$ to a homogeneous Hamiltonian, amounts to a \emph{tree tensor network} of isometries (albeit applied to a vector subspace): from a temporal top layer defined by a Hamiltonian containing only terms of rank $1$, one constructs the ground space of the full Hamiltonian $H$ by sequential applications of isometries with a simple topology.
We develop this observation further, and remark on implications for simulating the ground space of $H$, in \sectionref{sec:sampling}.
Note also that for a one-dimensional quantum chain and a sequential contraction, 
this construction gives rise to a {\it sequential preparation} of
a quantum state and hence to a {\it matrix-product state} of small bond dimension.

\subsection{The homogeneous case}
\label{sec:homogeneousCase}

Given a homogeneous Hamiltonian $H'$ (containing only terms of rank 1) acting on some system $V'$, consider a collection of 
vectors $\ket{\beta_{a,b}} \in \cH\sox{\ens{a,b}}$ such that
\begin{equation}
	h_{a,b} = \ket{\beta_{a,b}}\bra{\beta_{a,b}} \;.
\end{equation}
We interpret each two-spin Hamiltonian term $h'_{a,b}$ as a constraint on the corresponding two-spin marginal state $\rho_{a,b}$ of a state $\ket{\Phi} \in \ker(H)$, and attempt to obtain additional constraints on pairs of spins $a,b \in V'$ by combinations of the constraints which are already known.
Ref.~\cite{Bravyi06} shows that if $\ket{\Phi}$ lies in the kernel of $\bra{\beta_{a,b}}$ and $\bra{\beta_{b,c}}$ acting on the corresponding spins, it also lies in the kernel of the functional
\begin{gather}
	\label{eqn:constraintInduction}
		\bra{\beta'_{a,c}}
	 = 
		\bigl( \bra{\beta_{a,b}} \ox \bra{\beta_{b,c}} \bigr) \bigl(  \id \ox \ket{\Psi^-} \ox  \id \bigr)
\end{gather}
acting on the spins $a$ and $c$, where $\ket{\Psi^-} \propto \ket{0}\ket{1} - \ket{1}\ket{0}$ is the two-spin antisymmetric state vector.
We call such a constraint $\bra{\beta'_{a,c}}$ an ``induced'' constraint, and use the term \emph{induction of constraints} to refer to the operation on $\bra{\beta_{a,b}}$ and $\bra{\beta_{b,c}}$ which gives rise to $\bra{\beta'_{a,c}}$ in \eqref{eqn:constraintInduction}, up to a scalar factor.

For each induced constraint $\bra{\beta'_{u,v}}$ on a pair of spins $u,v$, we may add a term 
\begin{equation}
	\tilde h_{u,v} = \ket{\beta'_{u,v}}\bra{\beta'_{u,v}}
\end{equation}	
to the Hamiltonian $H'$, obtaining a Hamiltonian $\tilde H$ which (by construction) has the same kernel as $H'$.
If $H'$ already contains a term $h_{u,v}$ which is not colinear to the induced term $\tilde h_{u,v}$, these may be accumulated into a term $\bar h_{u,v}$ whose rank is at least $2$, and one may apply a two-spin contraction as described in \sectionref{sec:reductionsByIsometries}.
Otherwise, we may induce further constraints from the terms of $\tilde H$, until we obtain a \emph{complete} homogeneous Hamiltonian $H_c$: a Hamiltonian in which the two-spin constraints $\bra{\beta_{u,v}}$ are closed under constraint-induction.

By inducing constraints on pairs of spins, possibly performing two-spin contractions as in \sectionref{sec:reductionsByIsometries} when 
we obtain terms of rank $2$ or more, we may efficiently obtain a complete homogeneous Hamiltonian $H_c$ from a frustration-free, $2$-local Hamiltonian $H$.
Furthermore, Ref.~\cite{Bravyi06} shows that a complete homogeneous Hamiltonian $H_c$ acting on at least one spin 
(and which lacks single-spin operators \cite{Footnote2})
has a ground space which contains product states.
Thus, for homogeneous Hamiltonians $H$, we may either efficiently determine that it is frustrated, or efficiently obtain a Hamiltonian which is closed under constraint-induction.
In the latter case, we may construct product states in the kernel of $H$ by selecting states for each spin consistent with the two-spin constraints \cite{Footnote3}.

\section{Unfrustrated Natural Hamiltonians}
\label{sec:naturalHamiltonians}

We now present results concerning the ground state manifold of a ``physical'' class of $2$-local spin Hamiltonians.
We will say that a Hamiltonian $H$ is \emph{natural} if it is $2$-local, contains no isolated subsystems, 
and each term $h_{a,b}$ (acting on $\cH \ox \cH$) has at least one entangled excited state (i.{}e.{},~there exists an entangled state orthogonal to the ground state manifold of $h_{a,b}$).
Without loss of generality, we may further require that the ground energy of each term in $H$ is zero.
This is a natural assumption that typical physical interactions will satisfy: for instance, ferromagnetic or anti-ferromagnetic Ising interactions (which have excited eigenstates $\ket{0}\ket{1} - \ket{1}\ket{0}$ and $\ket{0}\ket{0} + \ket{1}\ket{1}$ respectively), ferromagnetic or anti-ferromagnetic XXX models (which also have those respective eigenstates), or indeed any interaction which is inequivalent to either $\diag(0,0,E_a,E_b)$ or $\diag(0,E_a,0,E_b)$ up to rescaling and a choice of basis for each spin.

Using the reductions of \sectionref{sec:reductionsByIsometries}, we show strong bounds on the dimension of the ground space of an unfrustrated natural Hamiltonian on spins.
This will allow us in \sectionref{sec:sampling} to describe a scheme for efficiently simulating the ground space of frustration-free natural Hamiltonians, and in \sectionref{sec:areaLaws} to demonstrate that the ground states of such Hamiltonians satisfy an entanglement area law.

\subsection{Ground-spaces of unfrustrated, natural, homogeneous Hamiltonians}
\label{sec:entangledCompleteHomogeneous}

We now present an extension of the analysis of Ref.~\cite{Bravyi06} for homogeneous and complete Hamiltonians $H_c$ (acting on a set $V_c$ of spins), to examine the ground-state manifold of $H_c$ in the case that $H_c$ is also natural.
We show, using techniques similar to those used in Ref.\ \cite[Section III A]{RandomQSat}, that the ground space of such a Hamiltonian is equivalent to the symmetric subspace $\Symm(\cH\sox{V_c}) \subset\cH\sox{V_c}$, up to some efficiently constructible choice of invertible operations on each spin.
As we may reduce more general Hamiltonians, i{.}e{.} having terms of rank $2$ or $3$ (extending beyond those Hamiltonians considered in Ref.\ \cite{RandomQSat}) to homogeneous natural Hamiltonians via the reductions of \sectionref{sec:reductionsByIsometries}, these results yield important consequences for natural frustration-free Hamiltonians in general.

Consider a Hamiltonian $H_c$ acting on $V_c$, where $H_c$ has no single-spin terms.
Because the two-spin constraints de- scribed by the terms of $H_c$ are closed under the induction of constraints (as described  by \eqref{eqn:constraintInduction}), and as there are no isolated subsystems, it is easy to show that every pair of spins is acted on by a non-zero term in $H_c$.
For such a Hamiltonian, the excited states $\ket{\beta_{a,b}}$ for the terms $h_{a,b}$ in $H_c$ are entangled states.
We may then construct a family of operators $\ens{L_v}_{v \in V_c} \subset \GL(2)$ such that
\begin{align}
	\label{eqn:localMappingToAsymmetricState}
		\bra{\beta_{u,v}}
	\;\propto\;
		\bra{\Psi^-}_{u,v} \bigl( L_u \ox L_v \bigr) 
\end{align}	
for each pair of spins $u,v$, where $\ket{\Psi^-}$ is again the two-spin antisymmetric state vector.
For instance, one may fix $L_a =  \id$ for an arbitrarily chosen spin $a \in V_c$, and determine linear operators $L_v$ satisfying \eqref{eqn:localMappingToAsymmetricState} for each $v \in V_c$ and operator $\bra{\beta_{a,v}}$.
Any such choice of operators $\ens{L_v}_{v \in V_c}$ satisfies \eqref{eqn:localMappingToAsymmetricState} for all $u,v$, which follows from the closure of the constraints $\bra{\beta_{u,v}}$ under induction:
\begin{align}
		\bra{\beta_{u,v}}
	\;\propto&\;\;
		\bigl( \bra{\beta_{ua}} \ox \bra{\beta_{a,v}} \bigr) \bigl(  \id \ox \ket{\Psi^-} \ox  \id \bigr)
	\notag\\[0.5ex]\propto&\;\;
		\bigl( \bra{\Psi^-} \ox \bra{\Psi^-} \bigr) \bigl( L_u \ox \ket{\Psi^-} \ox L_v \bigr)
	\notag\\[0.5ex]\propto&\;\;
		\bra{\Psi^-} \bigl( L_u \ox L_v \bigr).
\end{align}
We define scalars $\lambda_{u,v} \ne 0$ such that 
$\bra{\beta_{u,v}} = \lambda_{u,v} \bra{\Psi^-} \bigl(L_u \ox L_v\bigr)$ for each $u,v \in V_c$, 
and let $T = ( \bigotimes_{v \in V_c} L_v )\inv$: we then have
\begin{subequations}
\begin{align}
 	\mspace{-5mu}
 		\bigl( \bra{\beta_{u,v}} \ox  \id \bigr) T
	\;=&  
		\bra{\Psi^-}_{u,v} \ox
	T_{\overline{u,v}}	\;,
\end{align}
where we define $T_{\overline{u,v}}$ by
\begin{align}
		T_{\overline{u,v}}
	\;\;=&\;\;
		\lambda_{u,v} \! \bigotimes_{w\ne u,v} L_w^{-1} .
\end{align}
\end{subequations}
As the operators $T_{\overline{u,v}}$ have full rank, the operator $\big( \bra{\beta_{u,v}} \ox  \id \big) T$ then has the same kernel as $\bra{\Psi^-}_{u,v}$ , which is the $+1$-eigenspace of the $\textsc{swap}$ operator acting on $u$ and $v$.
It follows that the kernel of the Hamiltonian
\begin{align}
 		T\herm H_c T
	\;=&\;
		\sum_{u,v \in V_c} T\herm \bigl( \ket{\beta_{u,v}}\bra{\beta_{u,v}} \ox  \id \bigr) T
	\notag\\=&\;
		\sum_{u,v \in V_c} \ket{\Psi^-}\bra{\Psi^-}_{u,v}  \ox  T_{\overline{u,v}}\herm \:\! T_{\overline{u,v}}^{\phantom\dagger}
\end{align}
is the symmetric subspace $\Symm(\cH\sox{V_c})$; this corresponds to the result of \cite[Section III A]{RandomQSat}.

We remark on some important properties of $\Symm(\cH\sox{n_c})$, where $n_c = |V_c|$.
This subspace is spanned by uniform superpositions $\ket{W_k}$ of the standard basis states having Hamming weight $0 \le k \le n_c$,
\begin{gather}
		\ket{W_k}
	\;\;\;\propto\!
		\sum_{\substack{x \in \ens{0,1}^{n_c} \\ \!\!\! \|x\|_1  =  k}} \ket{x}	\;;
\end{gather}
thus $\dim \Symm(\cH\sox{n_c}) = n_c+1$.
This subspace may also be spanned by product state vectors 
$\ket{\alpha_0}\sox{n_c}\!,  \ldots,  \ket{\alpha_{n_c}}\sox{n_c}$ for any set of $n_c+1$ 
pairwise independent state vectors $\ket{\alpha_j} \in \cH$.
Thus, any natural Hamiltonian $H_c$ which is also complete and homogenous has a ground space of dimension $n_c+1$, and can be spanned by a family of classically efficiently simulatable state vectors
\begin{align}
		\ket{\Phi_j}
	\;=&\;
		\bigotimes_{v \in V_c} \bigl( L_v \ket{\alpha_j} \bigr),
\end{align}
for some choice of pair-wise independent single-spin state vectors $\ket{\alpha_0}, \ldots, \ket{\alpha_{n_c}} \in \cH$; we use this fact in \sectionref{sec:sampling}. Note that if even this efficient method should be too computationally costly for very large systems,
one can also Monte-Carlo sample from the ground state manifold in this way.

\subsection{Preservation of natural Hamiltonians under reductions}
\label{sec:reductionsPreserveNaturality}

A key feature of natural Hamiltonians (defined on \autopageref{sec:naturalHamiltonians}) is that the class of \emph{frustration-free} natural Hamiltonians on spins is preserved by the two-spin contractions described in \eqref{eqn:twoQubitIsometryInducedHamiltonian}.
This implies that the reductions of \sectionref{sec:reductionsByIsometries} map the ground-state manifold of an unfrustrated natural Hamiltonian $H$ provided as input to that of a complete, homogeneous, natural Hamiltonian; we may then apply the results of the preceding section to describe the ground-state manifold of $H$.

Consider an isometry $U_{u:u,v} : \cH\sox{\ens{u}} \to \cH\sox{\ens{u,v}}$ derived from a two-spin Hamiltonian term $h_{u,v}$ as described in \sectionref{sec:twoQubitIsometries}.
We may show that for any term $h_{a,u}$ in $H$, the corresponding term $h'_{a,u} = U_{u:u,v}\herm h_{a,u} U_{u:u,v}$ in the reduced Hamiltonian $H' = U_{u:u,v}\herm H U_{u:u,v}$ has an entangled excited state if the same holds for $h_{a,u}$.
We require the following two lemmas, whose proofs we defer to \hyperref[apx:techincalLemmas]{Appendix~\ref*{apx:techincalLemmas}}:

	\begin{lemma}[Product states]
		\label{lemma:nullOperatorProducts}
		For two-spin state vectors $\ket{\psi}$ and $\ket{\phi}$, we have 
		$\big( \bra{\psi} \ox  \id \big) \big(  \id \ox \ket{\phi} \big) = 0$ only if both $\ket{\psi}$ and $\ket{\phi}$ are product states.
	\end{lemma}
 
	\begin{lemma}[Product operators]
		\label{lemma:isometricContractionPreservingProducts}
		Let $U: \cH \to \cH \ox \cH$ be an isometry which is not a product operator.
		Let $\eta \ge 0$ be an operator on two spins, and \textup{$\eta' = \big(U\herm \ox \id_2) (\id_2 \ox \eta) (U \ox \id_2)$}.
		If $\eta'$ is not of full rank, then $\eta'$ is a product operator if and only if $\eta$ is a product operator.
	\end{lemma}

 	We show that frustration-free natural Hamiltonians are preserved by the reductions of \sectionref{sec:twoQubitIsometries} as follows.
	Let $H$ be a natural $2$-local Hamiltonian, and $h_{u,v}$ be a two-spin term in $H$.
	Define
	\begin{align}
		U_{u:u,v} \;=\; \ket{\psi_0}\bra{0}  +  \ket{\psi_1}\bra{1}
	\end{align}
	for orthonormal two-spin state vectors $\ket{\psi_0}, \ket{\psi_1}$ whose 
	span contains $\ker(h_{u,v})$; we require that $\ket{\psi_1}$ be entangled, which ensures that $U_{u:u,v}$ is not a product operator.
	Consider the terms $h'_{a,b} = U_{u:u,v}\herm  h_{a,b}  U_{u:u,v}$ which occur in the 
	Hamiltonian $H' = U_{u:u,v}\herm  H  U_{u:u,v}$.
	For any two-spin operator $h_{v,a}$ acting on $v$ and some other spin $a$, 
	the fact that $h_{v,a}$ has an entangled excited state implies in particular that it is not a product operator.
	Thus, $h'_{v,a}$ is a product operator only if it has full rank.
	If $H$ is frustration-free, $h'_{v,a}$ cannot have full rank; then $h'_{v,a}$ is not 
	a product operator, and in particular it will have entangled excited states.
	As $h'_{a,b} = h_{a,b}$ when $a,b \notin \ens{u,v}$, it follows that $H'$ is a natural 
	Hamiltonian; and as $H'$ has a kernel of the same dimension as $H$, it is 
	frustration-free as well.
	
	We may strengthen this result, to show that if $H$ is natural and frustration-free, 
	and also contains no two-spin terms of rank $1$, then the same is true of 
	$H' = U_{u:u,v}\herm  H  U_{u:u,v}$ as well.
	For any two-spin term $h_{v,a}$ acting on $v$ with rank at least $2$, consider 
	states $\ket{\varphi_0}, \ket{\varphi_1} \in \img(h_{v,a})$ such that 
	$\ket{\varphi_0} = \ket{\alpha}\ket{\beta}$ is a product state and $\ket{\varphi_1}$ is entangled; 
	Any subspace of $\cH \ox \cH$ of dimension at least $2$, such as 
	$\img(h_{u,v})$, contains a product state vector $\ket{\varphi_0}$; the existence of 
	$\ket{\varphi_1}$ is guaranteed by the definition of a natural Hamiltonian (compare also Ref.~\cite{Tarrach}).
	and choose real parameters $\lambda_0, \lambda_1 > 0$ such that
	\begin{equation}\label{psd}
		h_{v,a} - \lambda_0 \ket{\varphi_0}\bra{\varphi_0} - \lambda_1 \ket{\varphi_1}\bra{\varphi_1} \;\ge\; 0.
	\end{equation}
	Let $\eta_k = \ket{\varphi_k}\bra{\varphi_k}$ for $k \in \ens{0,1}$, and consider the images $\eta'_k$ under contraction by $U_{u:u,v}$ :
	\begin{align}
	 		\eta'_k
		\;=&\;\;
			U_{u:u,v}\herm  \eta_k  U_{u:u,v}
		\notag\\[1ex]=&\;\;
			\sum_{j,\ell} \bigl( \ket{j}\bra{\psi_j} \ox  \id \bigr) \bigl(  \id \ox \eta_k \bigr) \bigl( \ket{\psi_\ell}\bra{\ell} \ox  \id \bigr)
		\notag\\[0.5ex]=&\;\;
			\sum_{j,\ell} \ket{j}\bra{\ell} \ox M_{j,k} M_{\ell, k}\herm \;,
	\end{align}
	where we define $M_{j,k} = \big( \bra{\psi_j} \ox  \id \big) \big(  \id \ox \ket{\varphi_k} \big)$.
	By \autoref{lemma:nullOperatorProducts}, we have $M_{j,k} = 0$ only if both $\ket{\psi_j}$ and $\ket{\varphi_k}$ 
	are product operators; this implies that the operators $M_{1,k}$ in particular are non-zero, so that $\eta'_k \ne 0$ for any $k$.
	Note that
	\begin{align}
		\label{eqn:optorLowerBoundRank2}
			h'_{v,a}
		\;=&\; 
			U_{u:u,v}\herm  h_{v,a}  U_{u:u,v}
	\notag\\\ge&\; 
			\lambda_0   U_{u:u,v}\herm  \eta_0  U_{u:u,v} \;+\; \lambda_1   U_{u:u,v}\herm  \eta_1  U_{u:u,v}
	\notag\\=&\; 
			\lambda_0  \eta'_0 \;+\; \lambda_1  \eta'_1	\;;
	\end{align}
	because $H'$ has a non-trivial kernel, $h'_{v,a}$ has rank at most $3$, 
	in which case neither operator $\eta'_k$ has full rank.
	By \autoref{lemma:isometricContractionPreservingProducts}, $\eta'_1$ is not a product operator; as $\eta'_0$ is a product operator, these operators are linearly independent.
	Then, $\lambda_0 \eta'_0 \;+\; \lambda_1 \eta'_1$ has rank at least $2$; by \eqref{eqn:optorLowerBoundRank2}, the same is true of $h'_{v,a}$.
	If all of the terms in $H$ have rank $2$ or higher, the same then holds for $H'$ as well.

	Thus, if we apply the reductions of \sectionref{sec:reductionsByIsometries} to an initial Hamiltonian which is both natural and frustration-free, the resulting Hamiltonians will also be natural and frustration free.
	Furthermore, if $H$ contains no terms of rank $1$, then neither will the reduced Hamiltonians.
	Because the process of inducing constraints described in \sectionref{sec:homogeneousCase} also preserves the property of each term $h_{a,b}$ having entangled excited states, these invariants ensure that initial Hamiltonians with these properties (natural and frustation-free, and possibly containing no terms of rank $1$) may be reduced to homogeneous and complete Hamiltonians which have these same properties.
	We may then apply the results of \sectionref{sec:entangledCompleteHomogeneous} to these reduced Hamiltonians.

\begin{figure}[hbt]
\includegraphics[width=.55\columnwidth]{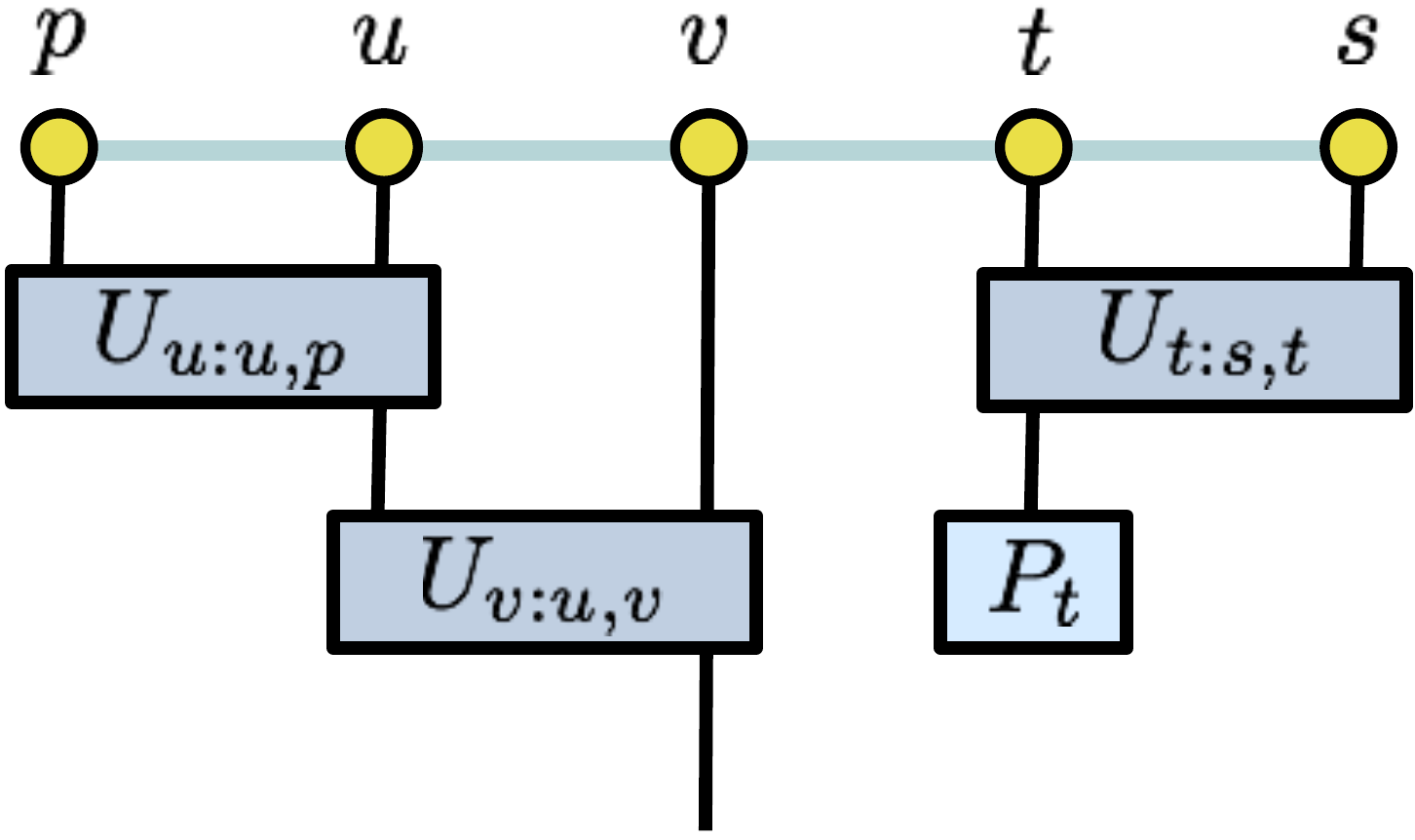}
\caption{A simple tree tensor network.}\label{fig:treeTensorIsometry}
\end{figure}

\section{Simulating ground spaces of frustration-free natural Hamiltonians}
\label{sec:sampling}

Building on the results of \sectionref{sec:naturalHamiltonians}, we now show how the reductions of \sectionref{sec:reductionsByIsometries} may be used to obtain a procedure for simulating states from the ground-state manifold of frustration-free natural Hamiltonians $H$.

\subsection{Tree tensor networks and matrix-product states}
\label{sec:simulateGroundSpace}

As we noted on \autopageref{remark:isometriesTreeTensorNetwork}, implicit in the reductions of \sectionref{sec:reductionsByIsometries} is that the isometric reduction from general Hamiltonians to homogeneous instances has the form of a tree-tensor network.
Thus, simulating the ground state manifold of any unfrustrated $2$-local Hamiltonian on spins may be reduced to that of a complete homogeneous Hamiltonian acting on a smaller system.
In this section, we sketch this reduction.

For any unfrustrated $2$-local Hamiltonian $H$ on $n$ spins, we may apply two-spin reductions as described in \sectionref{sec:reductionsByIsometries} until we obtain a homogeneous instance without single-spin terms.
We then attempt to induce additional constraints via \eqref{eqn:constraintInduction}, and apply further two-spin reductions if we obtain terms of rank $2$ or $3$.
If $H$ is frustration-free, this process will ultimately terminate in a complete homogeneous Hamiltonian $H_c$ on a subset $V_c \subset V$.

Consider the tensor network $T$ which performs the complete reduction as above; we describe $T$ in reverse order, as introducing new spins to represent a unitary embedding of $\ker(H_c)$ into $\ker(H)$.
The various “spin contraction” isometries $U_{u:u,v}$ as in \eqref{eqn:twoQubitIsometry} each have a single input-index and two output indices; the “spin deletion” isometries $P_v$ have no input indices at all.
These are applied sequentially, giving rise to an acyclic directed network.
As the in-degree of each tensor is at most $1$, it follows that the network contains no cycles at all (neither directed nor undirected): the output indices of each tensor represent spins whose state depends on only a single spin at the input.
Put another way: any spin $v$ which is introduced by an isometry $U_{u:u,v}$ may be considered a ``daughter spin'' of a unique parent $U$, which imposes a tree-like hierarchy 	
on the tensor network $T$, as illustrated in \autoref{fig:treeTensorIsometry}.
Strictly speaking, the quantum circuit or tree tensor network will have the structure of a \emph{forest} graph, which is a graph which may have more than one connected component, each of which are trees.

The roots of each tree are spins $u$ which are either prepared by an isometry $P_u$ derived from the removal of 
single-spin terms, or which correspond to free indices at the input of the tensor network $T$.

In the case that $H$ is non-degenerate, the resulting tensor network will (by that fact) simply be a tree-tensor network with no free input indices. 
Conversely, if the input Hamiltonian $H$ is degenerate, there will necessarily be free input indices, representing a domain consisting of a state space of dimension at least $2$.
In the latter case, the tensor network $T$ will yield ground states of the original unfrustrated Hamiltonian $H$ if and only if it operates on a state $\ket{\varphi} \in \ker(H_c)$ at the input, where $H_c$ is the complete homogeneous instance obtained by the Hamiltonian reductions.
Thus, if one may efficiently simulate states from the ground space of such a Hamiltonian, we may apply the network $T$ to simulate the ground space of the original Hamiltonian $H$.

\subsection{Efficiently simulating ground spaces of unfrustrated natural Hamiltonians}

Tensor networks $T$ with free input indices, and with a tree-like structure such as described above, can be efficiently simulated over inputs with low Schmidt measure~\cite{SM}, as follows.

For any observable $\Omega$ acting on $m$ spins, one may evaluate $\anglebr{\Omega}_H$ for the maximally mixed state over the ground-state manifold of $H$ by computing the expectation $\anglebr{T\herm \Omega T}_{\smash{H_c}}$ over the ground-state manifold of the complete homogeneous Hamiltonian $H_c$ obtained as described in \sectionref{sec:reviewBravyi}.
As the tensor $T$ has tree-stucture, the observable $T\herm \Omega T$ also acts on at most $m$ spins.
If we can obtain an orthonormal basis for $\ker(H_c)$ which may be succinctly described in terms of product states, we may evaluate expectation values of $T\herm \Omega T$ with respect to $m$-fold products of single spin states.

As we noted in \sectionref{sec:entangledCompleteHomogeneous}, $\ker(H_c)$ can be spanned by a collection of $n_c + 1$ product vectors (where $n_c$ is the number of spins on which $H_c$ acts).
Let $\ket{\Phi'_0}, \ket{\Phi'_1}, \ldots, \ket{\Phi'_{n_c}} \in \ker(H_c)$ be a collection of independent product vectors,
\begin{align}
		\ket{\Phi'_j}
	\;=\;
		\ket{\varphi_{j,1}} \ox \ket{\varphi_{j,2}} \ox \cdots \ox \ket{\varphi_{j,n_c}}	\;.
\end{align}
We may efficiently compute a projection of $\Omega_c$ onto $\ker(H_c)$ by performing a suitable transformation of the matrix
\begin{align}
		\label{eqn:skewRestriction}
		W(\Omega_c)
	\;=&\;
		\sum_{j,k = 0}^{n_c} \ket{j}\bra{\Phi'_j} \Omega_c \ket{\Phi'_k}\bra{k}	\;,
\end{align}
as follows.
The operator $W(\idop)$ in particular is positive definite; we thus have $W(\idop) = U \Delta U^\dagger$ for some unitary $U$ unitary and positive diagonal matrix $\Delta$.
It is not difficult to show that
\begin{align}
			\Delta^{-1/2}  U^\dagger \sum_{j = 0}^{n_c} \ket{j}\bra{\Phi'_j}
		\;=\;
			\sum_{j = 0}^{n_c} \ket{j}\bra{\Phi_j}	\;,	
\end{align}
for some orthonormal basis $\ket{\Phi_0}, \ldots, \ket{\Phi_{n_c}}$ of $\ker(H_c)$, by taking the product of the above operator with its adjoint.
Thus, the restriction of $\Omega_c$ to $\ker(H_c)$ with respect to the basis of states $\ket{\Phi_j}$ may be computed as
\begin{align}
		\bar\Omega
	\;=&\;
		\Delta^{-1/2} U^\dagger W(\Omega_c) U \Delta^{-1/2}	.
\end{align}
By considering operators $\Omega_c = T^\dagger \Omega T$, this allows us to compute the restriction of operators to the ground-space of $H_U$: see \autoref{fig:isometricSampling}.

\begin{figure}[hbt]
\includegraphics[width=.5\columnwidth]{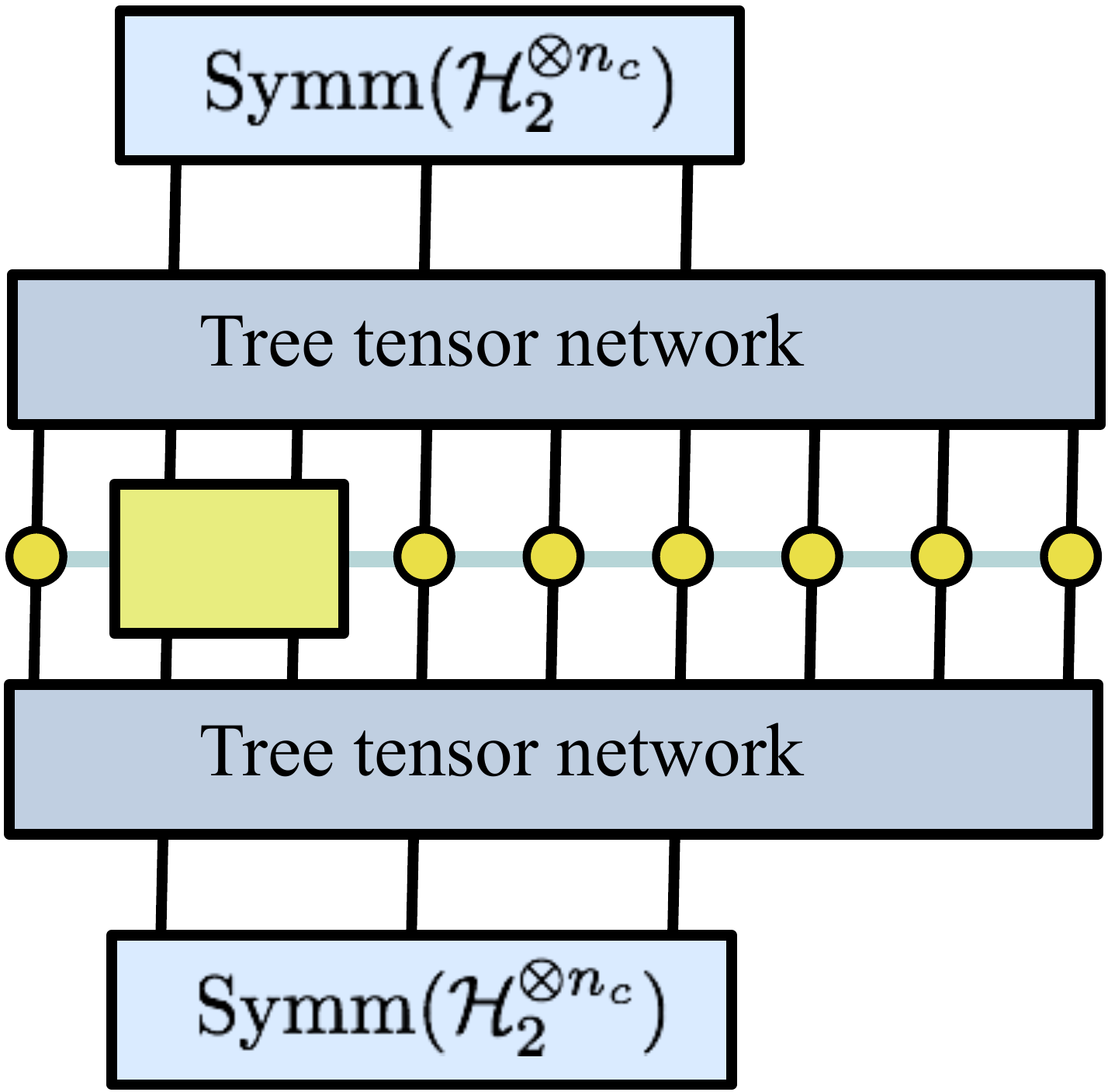}
\caption{Schematic diagram showing the isometric decomposition for efficient simulating the 
ground state manifold. A local Hamiltonian term supported on two sites
is marked yellow.}\label{fig:isometricSampling}
\end{figure}

We may thus efficiently estimate such observables with respect to ground-states of $H_U$: for constant $m$, the required inner products may be calculated as (sums of) scalar products of at most $n_c$ inner products over vector spaces of bounded dimension.
To evaltuate the value of $\Omega$ with respect to ground states of the input Hamiltonian $H$, it suffices to analyse the polynomially-sized operator $\bar \Omega$ representing the action of $\Omega$ on the ground-state manifold.

\section{Entanglement bounds for ground states of frustration-free natural Hamiltonians}
\label{sec:areaLaws}

The fact that the reductions of \sectionref{sec:reductionsByIsometries} preserve the class of natural Hamiltonians (as defined on \autopageref{sec:naturalHamiltonians}) allow us also to make more global statements about ground states for frustration-free Hamiltonians, again by reduction to the complete homogeneous case described in \sectionref{sec:entangledCompleteHomogeneous}.
In this section,  considering frustration-free natural Hamiltonians $H$, we demonstrate an area law for the entanglement possible in a ground state of $H$ between any subsystem $A \subset V$ and its environment $V \setminus A$.
We also consider some very general cases in which still stronger upper bounds on the entanglement may be obtained.

\subsection{Area law for frustration-free natural Hamiltonians}

We consider first the case where $A$ is a \emph{contiguous} subsystem (i{.}e{.} for which there is a path between any pair of spins in $A$, following edges in the interaction graph of the Hamiltonian $H$), and subsequently generalize the observation in this case to arbitrary subsystems $A$.
We first decompose
\begin{align}
		H
	\;=\;
		H_A + H_B + H_{A,B}	\;,
\end{align}
for $B = V \setminus A$, and where $H_A$ and $H_B$ contain all terms internal to $A$ and $B$, respectively.
We then apply the reductions of \sectionref{sec:reductionsByIsometries} to the subsystem $A$.
That is, we perform two-spin contractions as described in \sectionref{sec:twoQubitIsometries} for any two-spin terms in $H_A$ of rank $2$ or $3$, and perform spin-deletions as described in \sectionref{sec:oneQubitRemoval} for any single-spin terms in $H_A$.
Performing the constraint-induction process of \eqref{eqn:constraintInduction} --- again on the terms acting on $A$ alone --- and then reducing further reductions as necessary, we eventually obtain a Hamiltonian $\tilde H$ of the form
\begin{align}
	\label{eqn:reducedBipartiteHamiltonian}
		\tilde H
	\;=\;
		\tilde H_{\tilde A} + H_B + \tilde H_{\tilde A, B}	\;,
\end{align}
where $\tilde A \subset A$ is the set of spins remaining after the reduction process, and where $\tilde H_{\tilde A}$ is a homogeneous and complete Hamiltonian. The Hamiltonians $\tilde H_{\tilde A}$ and $\tilde H_{\tilde A,B}$ contain the terms derived from $H_A$ and $H_{A,B}$ respectively by the reduction process. 
In other words, we perform a complete tree tensor reduction on the subsystem $A$, until we obtain a Hamiltonian whose restriction to $A$ is homogeneous and complete.
As $\dim\ker(\tilde H) = \dim\ker(H) > 0$, we have $\ker(\tilde H_{\tilde A}) > 0$ and $\ker(H_B) > 0$ as well; in particular, $\ker(\tilde H) \subset \ker(\tilde H_{\tilde A})$, so that $\tilde H_{\tilde A}$ is frustation-free.
Because $\ker(\tilde H) \subset \ker(H_B)$ as well, we then have
\begin{align}
		\label{eqn:tensorDecomposeKernels}
		\ker(\tilde H)	\subset \ker(\tilde H_{\tilde A}) \ox \ker(H_B)
\end{align}
taking the restrictions of $\tilde H_{\tilde A}$ and $H_B$ to their respective subsystems $\tilde A$ and $B$.
Let $\tilde n = |\tilde A|$: as $\tilde H_{\tilde A}$ is also homogeneous and complete, it has kernel of dimension $\tilde n + 1$ by \sectionref{sec:entangledCompleteHomogeneous}.

In the case where $A$ contains multiple components $A_1, A_2, \ldots, A_k$ with respect to the interaction
graph of $H$ (where each $A_j$ is disconnected from the others but connected internally), we may perform the Hamiltonian reductions of \sectionref{sec:reductionsByIsometries} to each component independently.
We may further decompose the Hamiltonian $\tilde H_{\tilde A}$ obtained in \eqref{eqn:reducedBipartiteHamiltonian} as
\begin{align}
		\tilde H_{\tilde A}
	\;=\;
		\tilde H_{\tilde A_1}
		+ 
		\cdots
		 +
		\tilde H_{\tilde A_k}	,
		\;\;\;
			\text{where $\tilde A_j = \tilde A \inter A_j$}.
\end{align}
As $\tilde H_{\tilde A}$ is unfrustrated, each of the sub-Hamiltonians $\tilde H_{\tilde A_j}$ is unfrustrated as well, in which case we may write
\begin{align}
		\label{eqn:tensorDecomposeKernels-b}
		\ker(\tilde H_{\tilde A})
	\;\subset\;
		\ker(\tilde H_{\tilde A_1})	\ox	\cdots\ox	\ker(\tilde H_{\tilde A_k}),
\end{align}
similarly to \eqref{eqn:tensorDecomposeKernels}.
Then, the dimension of $\ker(H_{\tilde A})$ is bounded by the product of $\dim\ker(\tilde H_{\tilde A_j}) = \tilde n_j + 1$ for each subsystem, where $\tilde n_j = \card{\smash{\tilde A}_j}$.
Let $\alpha_j$ be the number of spins in $A_j$ which are adjacent in the spin lattice to spins in $B$, and let $\tilde \alpha_j \le \alpha_j$ be the number of such spins in $\tilde A_j$. 
For any $H$ with nearest neighbor interactions on a lattice in finitely many dimensions (in the graph theoretical sense),
there exist scalars $c, K > 0$ such that
\begin{equation}
	\label{eqn:constantDimension}
	\alpha_j \ge K n_j^c
\end{equation}	
for each subsystem $A_j$.
We may then bound on the dimension of $\ker(\tilde H_{\tilde A})$ in terms of these ``boundary spins'' as
\begin{align}
		\log(\dim & \ker(\tilde H_{\tilde A}))
	 \;\le\;
		\sum_{j = 1}^k
			\log(\dim \;\! \ker(\tilde H_{\tilde A_j}))
	\notag\\\le&\;
		\sum_{j = 1}^k
			\log(\tilde n_j + 1)
	\;\le\;
		\sum_{j = 1}^k
			\tilde n_j^c
	 \;\le\; 
		\frac{\alpha}{K},	
\end{align}
where $\alpha = \alpha_1 + \cdots + \alpha_k$ is the number of spins in $A$ adjacent to elements of $B$.

From this bound on $\dim\ker(\tilde H_{\tilde A})$, it follows that  any pure state $\ket{\Phi}$ in the ground space of $\tilde H$ 
has Schmidt measure~\cite{SM} at most ${\alpha}/{K}$ across the partition $\tilde A + B$, and so can support at most this many e-bits of entanglement between 
$\tilde A$ and $B$.
Because any state vector $\ket{\Psi} \in \ker(H)$ can be obtained from some vector $\ket{\Phi} \in \ker(\tilde H)$ by a 	
network of isometries acting only on spins in $A$,  it follows that any state vector $\ket{\Psi}$  in the ground-state manifold of $H$ contains at most ${\alpha}/{K}$ 
e-bits of entanglement between $A$ and $B$. Note the similarity to ground states of low Schmidt rank close to factorizing ground states in 
Heisenberg models \cite{Factorization}.
In case the ground state $\rho$ is degenerate, each pure state in the spectral decomposition of $\rho$
will have that property. As the entanglement of formation is convex (this usually being taken as a necessary
property of an entanglement monotone), one obtains the bound 
\begin{equation}
	E_F(\rho) \le \frac{\alpha}{K}
\end{equation}	
for the entanglement of formation between $A$ and $B$.
Finally, let $\partial A$ be the set of edges between $A$ and $B$.
By definition, for each edge in $\partial A$, there is a spin in $A$ which is adjacent to some spin in $B$; then we have $\alpha \le |\partial A|$, so that
\begin{align}
	E_F(\rho) \le \frac{|\partial A|}{K} .
\end{align}
Thus, the amount of entanglement which can be supported by a ground state of $H$ between $A$ and $B$ is 
governed by an area law. 
We summarize:

\begin{proposition}[Area law]
\label{propn:areaLaw}
Let $H$ be an unfrustrated natural
Hamiltonian on a lattice $V$, and denote with $\rho$ its (possibly degenerate)
ground state. Then for any subsystem $A\subset V$, the 
entanglement of formation of $\rho$
with respect to $A$ and $V\setminus A$ satisfies an
area law, i.e. there exists a constant $C>0$ of the lattice model such that
\begin{equation}
	E_F(\rho) \le C |\partial A|.
\end{equation}
\end{proposition}

\subsection{Stronger entanglement bounds for contiguous subsystems}

The above analysis imposes no additional constraints, beyond the requirement that $H$ be natural and frustration-free.
We may obtain still stronger bounds --- by the logarithm of the system size, or even by a \emph{constant} --- on the entanglement between $A$ and its environment $B = V \setminus A$, under fairly general conditions on the subsystem $A$ when it is a contiguous subsystem.

\subsubsection{Contiguous subsystems in general}

Implicit in the analysis of the previous subsection is a stronger entanglement bound for contiguous subsystems in general: we observe that if $A$ consists of a single component, we have
\begin{align}
		\dim\ker(\tilde H_{\tilde A})
	\;=\;
		\tilde n + 1
\end{align}
for $\tilde n = |\tilde A|$ by the analysis of \sectionref{sec:entangledCompleteHomogeneous} (where $H_{\tilde A}$ is the reduced Hamiltonian acting on the subsystem $A$ described in \eqref{eqn:reducedBipartiteHamiltonian}).
By a similar analysis, if $\alpha$ is the number of spins in $A$ adjacent to at least one spin in $B$, we may use \eqref{eqn:constantDimension} to obtain
\begin{align}
		\log(\dim & \ker(\tilde H_{\tilde A}))
	 =
		\log(\tilde n + 1)
	\le
		\frac{\log(\alpha / K + 1)}{c} \:\!.
\end{align}
As $\alpha \le |\partial A|$, we may then obtain:

\begin{proposition}[Logarithm law for contiguous systems]
Let $H$ be an unfrustrated natural Hamiltonian on a lattice $V$, and denote with $\rho$ its (possibly degenerate)
ground state.
There then exists a constant $C > 0$ of the lattice model such that, for any \emph{contiguous} subsystem $A\subset V$, the entanglement of formation of $\rho$ with respect to $A$ and $V\setminus A$ satisfies
\begin{equation}
	E_F(\rho) \le C \log|\partial A|.
\end{equation}
\end{proposition}

\subsubsection{Subsystems acted on by many high-rank Hamiltonian terms}

In the above result, we have neglected the difference in the sizes of the subsystem $A$, and the reduced subsystem $\tilde A \subset A$.
The difference in their sizes will be precisely the number of isometric reductions performed to obtain $\tilde A$ from $A$.
Each isometric reduction corresponds to either an \emph{edge contraction} in the interaction graph $G$ of  the Hamiltonian $H$, or a \emph{vertex deletion} in $G$, yielding an interaction graph $G'$ for the Hamiltonian $H'$.
Two-spin isometries $U_{u:uv}$ represent the reduced state-space of the two-spin subsystem $\ens{u,v}$ as the image of a single spin under an isometry: the corresponding reduction may thus be represented as a ``contraction'' of two spins into one, as illustrated in \autoref{fig:graphReduction}.
Single-spin terms $h_u$ may be represented by loops on vertices: isometric reductions arising from terms $h_{u,v}$ of rank $3$ also yield a loop on the contracted vertex.
Spin-removal reductions $P_u$ correspond to the deletion of a vertex $u$ with a loop, which removes all edges $au$ incident to $u$, possibly replacing them by loops on the neighbors $a$.

This representation of Hamiltonian reductions in terms of graphs is underdetermined, in that it is not always possible to determine the ranks of the reduced Hamiltonian $H'$ from those of the Hamiltonian $H$ prior to contraction.
However, the correspondence to graph reductions motivates a simple observation.
Consider a subsystem $A$, and consider the Hamiltonian $H_A$ together with its interaction graph $G_A$.
We may ``colour'' or ``rank'' the edges of $G_A$ according to whether the term corresponding to each edge is rank-$1$ 
(which we call ``light'' edges) or has rank $2$ or $3$ (which we call ``heavy'' edges).
The two-spin isometric reductions of \sectionref{sec:twoQubitIsometries} 
required to obtain $\tilde H_{\tilde A}$ corresponds to contractions of all heavy edges in $G_A$.
As such contractions preserve connectivity, this implies that the interaction graph $\tilde G_{\tilde A}$ 
corresponding to $\tilde H_{\tilde A}$ has as many vertices as there are connected components in the ``heavy subgraph'' of $G_A$.
In particular, if the number of ``heavy'' connected components (components connected only by heavy edges)  is bounded above by some parameter $\beta$, we then obtain $\tilde n = |\tilde A| \le \beta$, so that
\begin{align}
		\log(\dim\ker(\tilde H_{\tilde A})) \le \log(\beta + 1).
\end{align}
A consequence of this is that if $H$ is a frustration-free natural Hamiltonian which contains only terms of rank $2$ or $3$, all edges in $G_A$ will be heavy, so that it consists of a single heavy component; we then have $\log(\dim\ker(\tilde H_{\tilde A})) = 1$.
In this case, there is at most one e-bit of entanglement between $A$ and any other, disjoint subsystem in the lattice.

We may further refine this observation by considering the impact of Hamiltonian terms of rank $3$.
More generally, we may consider rank-$1$ single-spin terms in the reduced Hamiltonians, arising either from rank-$1$ 
terms in preceding Hamiltonians, or from performing an isometric reduction on terms of rank $3$.
Such single-spin terms correspond to loops on vertices in the interaction graphs $G'_{A'}$ of the reduced Hamiltonians $H'_{A'}$.
In the case of a frustration-free natural Hamiltonian containing such terms, we may show that the Hamiltonian $H$ is non-degenerate with a ground state consisting essentially of a product of single-spin states (together with a single two-qubit entangled state if the original Hamiltonian $H$ contains a term of rank $3$).
Consider the effect of preferentially performing single-spin deletions in the process of reducing $H$ by isometries: for a natural Hamiltonian, we may easily verify that removal of such a vertex $u$ (i{.}e{.} performing the single-spin removal reduction of \sectionref{sec:oneQubitRemoval}) will induce loops corresponding to single-spin terms on all neighbors of $u$.
These spins may then be removed in turn, inducing still further loops; by the requirement that the original interaction graph $G$ be connected, this ultimately results in the removal of every spin on which $H$ acts, each by independent single-spin isometries which describe a 
fixed single-qubit state vector $\ket{\varphi_u}$.
As a result, the \emph{entire lattice} contains no entanglement, or at most one e-bit if $H$ contained a single rank-$3$ term giving rise to a ``seed'' loop; any subsystem $A$ which is not acted on by the rank-$3$ term therefore contains no entanglement, nor has any entanglement with its environment.

The above exhibits the fragility of the condition of  frustration-freeness: it follows, for instance, that any natural Hamiltonian $H$ which contains as many as two terms of rank $3$ is necessarily frustrated (i{.}e{.} does not have a ground space characterized by those of its interaction terms).
Because the same unique ground state must be produced by any reduction, e{.}g{.} in which we first perform two-qubit isometries, it follows that each two-qubit isometry in such a reduction must also map the single-spin states (describing the unique ground state of the reduced Hamiltonians) to 
product states, which is of course highly unlikely if instead one considers arbitrarily chosen two-qubit 
isometries and single-product input states.
These observations may be used together with the random satisfiability results of Ref.~\cite{RandomQSat} 
to suggest that ``exact'' frustration-freeness is likely to be rare in physical systems; small perturbations are likely to cause frustration.
This is nothing but a manifestation of a fragility against spontaneous symmetry breaking.
However, in \sectionref{sec:almostUnfrustrated}, we suggest ways in which systems which differ only slightly from frustation-free systems may be examined using the techniques of \sectionref[Sections]{sec:sampling} and\sectionref[]{sec:areaLaws}.

\section{Different models of frustation-free  Hamiltonians}

In this section, we consider frustration-free Hamiltonians $H$, but suspend our earlier restriction to natural Hamiltonians (as described on \autopageref{sec:naturalHamiltonians}) in order to consider different models of Hamiltonians that are of interest.
In doing so, we will compare the resulting analysis to the case of frustration-free natural Hamiltonians in \sectionref{sec:areaLaws}.

\subsection{Rank-two terms lacking entangled excited states}
\label{sec:rank2unnatural}

Any Hamiltonian term $h_{a,b}$ in $H$ which has rank $2$ and has only product states orthogonal to its 
ground space is of the form $h_{a,b} = \ket{\varphi}\bra{\varphi} _a \ox \eta_b$ (or the reverse tensor product), where 
$\eta$ is a single-spin operator of full rank and $\ket{\varphi} \in \cH$.
This operator has the same kernel as the single-spin operator $\ket{\varphi}\bra{\varphi}_a \ox \id_b$; therefore, if $H$ 
is frustration-free, we may perform this substitution without any change to the ground state manifold or its properties.
As any rank-$3$ operator has entangled states orthogonal to its (unique) ground state, we may therefore restrict to the 
case where ``non-natural'' terms $h_{a,b}$ occuring in $H$ have rank $1$, so long as we permit input 
Hamiltonians with single-spin terms.

\subsection{Unfrustrated translationally invariant Hamiltonians}
\label{sec:unfrustratedTI}

Consider a frustration-free Hamiltonian $H$ in which the interaction terms $h_{a,b}$ of each spin $a$ is the same for all of its neighbors $b$.
If $H$ is not natural, it follows that $h_{a,b} = \ket{\alpha}\bra{\alpha}_a \ox \ket{\beta}\bra{\beta}_b$ for some 
states $\ket{\alpha}, \ket{\beta} \in \cH$; and by a suitable choice of basis on each site, we may without loss 
of generality let $\ket{\alpha} = \ket{\beta} = \ket{1}$.

Consider a ground state vector $\ket{\Phi}$ of the Hamiltonian.
For each site $a$, if the state vector of $\ket{\alpha}$ is not given by $\ket{1}$, it follows that all of the neighbors of $a$ are in 
$\ket{1}$; and conversely, if all of the neighbors of some site $a$ are in $\ket{1}$, the site $a$ may be in an arbitrary 
single-spin state without contributing to the energy of the global state.
It follows that the ground-state manifold of $H$ consists of all superpositions of product states in which all sites are in 
 $\ket{1}$, except for some set of mutually non-adjacent sites $A \subset S$, whose spins may have arbitrary states  (including states which are entangled with other sites in $S$).
In particular, for bipartite lattices, this includes states in which the entire ``even'' sublattice of sites an even distance from the origin are in 
 $\ket{1}$, and the opposite ``odd'' sublattice may have an arbitrary entangled state.

Thus, if $H$ is isotropic and frustration-free, then without loss of generality it is either natural, or contains subspaces in which large subsystems of the lattice are essentially unconstrained, and may occupy states with arbitrarily large entanglement content.
Consequently, one may expect that unfrustrated translationally-invariant Hamiltonians should have interaction terms with entangled excited states, i{.}e{.} be given by natural Hamiltonians.

\subsection{Unfrustated lattices with randomly located product terms and percolation}

Finally, we wish to consider a class of \emph{random} Hamiltonians which includes non-natural Hamiltonians, and compare the 
behaviour of their ground-state manifolds to natural Hamiltonians.
If one distributes random Hamiltonian terms $h_{a,b}$ over nearest-neighbor pairs in an arbitrary lattice, then they will almost certainly have an entangled excited state, as the highest-energy eigenstate of each term will be a product state with probability zero.
This remains true even if one constrains each interaction term $h_{a,b}$ in the lattice to have ranks described by integers $r_{a,b} \in \ens{1,2,3}$ selected according to any distribution, including the case where every term has rank $1$.
In order to obtain a random model of non-natural Hamiltonians, we must explicitly designate certain interactions $h_{a,b}$ to be rank-$1$ product operators (non-natural terms of higher rank being subject to the remarks of \sectionref{sec:rank2unnatural} above), and consider the scaling of the resulting lattice model.

Consider a $d$-dimensional rectangular lattice, in which each term $h_{a,b}$ has rank-$1$, and for each term we randomly determine whether $h_{a,b}$ is a product term (i{.}e{.} satisfies $h_{a,b} = \ket{\alpha}\bra{\alpha}_a \ox \ket{\beta}\bra{\beta}_b$ for some $\ket{\alpha},\ket{\beta} \in \cH$) or an entangled term (satisfies $h_{a,b} = \ket{\gamma}\bra{\gamma}$ for some entangled $\ket{\gamma} \in \cH \ox \cH$).
The probability that $h_{a,b}$ is entangled is given by some fixed $p > 0$, independently for each edge.
Having determined whether $h_{a,b}$ is entangled or not, we select a random rank-$1$ projector for $h_{a,b}$ subject to that constraint on $h_{a,b}$.
Considering only frustration-free Hamiltonians $H$ constructed under such a model \cite{Footnote4}
and subsystems $A$ of the lattice on which $H$ acts, we wish to determine the dimension of the ground-state manifold $M = \ker(H_A)$; by a similar analysis as in \sectionref{sec:areaLaws}, this will indicate how close the ground-states of $H$ come to obeying an entanglement area law.

As we noted in \sectionref{sec:entangledCompleteHomogeneous}, the process of inducing rank-$1$ constraints as in \eqref{eqn:constraintInduction} will yield entangled (``natural'') constraints from two other entangled constraints.
Consider the subgraph of the lattice consisting of entangled constraints: it follows that any subsystem $A \subset V$ of the lattice which is connected only by entangled constraints forms a subsystem for which $H_A$ is a natural Hamiltonian, with a kernel of dimension at most $|A| + 1$.
Conversely, we may easily show that for any product term $h_{a,u} = \ket{\alpha}\bra{\alpha}_a \ox \ket{\beta}\bra{\beta}_u$, the constraints $\tilde h_{a,v}$ induced by $h_{a,b}$ together with any other constraint $h_{u,v}$ will also be a product term (regardless of whether $h_{u,v}$ is a product term).
Thus, such product terms $h_{a,u}$ in the Hamiltonian represent obstacles to the induction of constraints which would yield bounds on entanglement: as we noted in \sectionref{sec:unfrustratedTI} above, the prevalence of product terms in a Hamiltonian $H$ allow for the effective decoupling of large subsystems in the ground-state manifold of $H$, yielding extremely high degeneracy.

These observations suggest an approach to bounding $\dim(M)$ using \emph{percolation theory}~\cite{Grimmett} to bound the number and size of components connected by entangled edges in a large convex subset (for a review on applications of percolation theory in quantum information, see Ref.\
\cite{Kieling}).
We may consider the worst case scenario in which no additional constraints may be induced between any two subregions $A_1, A_2$ which are internally connected by entangling terms, but separated by a barrier of product terms which effectively decouple the subsystems $A_1$ and $A_2$.
If the probability $p$ is above the percolation threshhold $p_c$ of the lattice, we may apply the following results:

\begin{proposition}[{\cite[Theorem 4.2]{Grimmett}}]
	\label{prop:countClusters}
	Let $A$ be a hypercube consisting of $n$ vertices, in a $d$-dimensional rectangular lattice with edge-percolation probability $p$.
	Then there exists a positive real $\kappa_p \in \R$ such that the number of connected components in $A$ grows as $\kappa_p n$, as $n \to* \infty$.
\end{proposition}
\begin{proposition}[{\cite[Theorem 8.65]{Grimmett}}]
	\label{prop:sizeClusters}
	Let $C$ be a finite-size connected component containing an arbitrary vertex (e{.}g{.} the origin) in a $d$-dimensional rectangular lattice with edge-percolation probability $p$.
	For $p > p_c$, there exists a positive $\eta_p \in \R$ such that
	\begin{align}
		\Pr_{p}\big(|C| = s\big) \;\le\; \exp\Big( \!-\eta_p  s^{(d-1)/d} \Big).
	\end{align}
\end{proposition}
Both $\kappa_p$ and $\exp(-\eta_p)$ in the propositions above are analytic for $p > p_c$, and thus must converge to $0$ in the limit $p \to* 1$.
We may thus describe an upper bound on the dimension of $M$ as follows, for $A$ a large cube containing $n \gg 1$ spins.
For $p > p_c$, there is almost surely a unique maximum-size component $A_0$ of the lattice which is connected by entangled edges: because the percolation probability $\theta_p$ is strictly positive (by definition) for $p > p_c$, we will have $|A_0| = \theta_p n$ on average.
Each subsystem $A_0, A_1, \ldots, A_k \subset A$ which is connected by entangled edges induces a natural Hamiltonian $H_{A_j}$ which has a kernel of dimension at most $|A_j| + 1$: we may bound $\dim(M)$ by noting that
\begin{align}
		M
	\subset
		\ker(H_{A_0}) \ox \ker(H_{A_1}) \ox \cdots \ox \ker(H_{A_k})\;,
\end{align}
as in \eqref{eqn:tensorDecomposeKernels-b}.
This allows us to obtain the bound
\begin{align}
		\log\dim(M)
	\;\;\le&\;\;
		\sum_{j \ge 0} \; \log\dim(\ker(H_{A_j}))
	\notag\\[0.5ex]\le&\;\;
		\sum_{j \ge 0} \; \log\big(|A_j| + 1\big).
\end{align}
By \autoref{prop:countClusters}, the expected number of components $k$ grows like $\kappa_p n$ for some $\kappa_p > 0$ as $n \to* \infty$; using the probability bound on the typical finite component size of \autoref{prop:sizeClusters} as the probability of an indistinguished connected component having size $s$, we obtain the upper bound
\begin{align}
		\Exp_p &\Big[\log\dim(M)\Big]
\notag\\[0.5ex]\le&\;
		\log\big(|A_0| + 1) + \sum_{j \ge 1} \log\big(|A_j| + 1\big)
\notag\\[0.5ex]\le&\;
		\log(\theta_p n + 1) +\;
		\sum_{j = 1}^{\kappa_p n} \left[
		\sum_{s = 1}^{\infty}
			\frac{\log(s+1)}{\exp\big(\eta_p  s^{(d-1)/d}\big)} 	\right]
		\notag\\[0.5ex]\le&\;\;
				\log(\theta_p n + 1) + \kappa_p C_p n\;, 
\end{align}
where $C_p$ is the sum in square brackets (which is small for $\exp(-\eta_p)$ small).

As $\log\dim(M)$ is also the logarithm of the maximum Schmidt rank of any state with respect to the bipartition into $A$ and $V \setminus A$ for the lattice $V$, the amount of entanglement scales with the logarithm of the size of the cube $A$, with a small and tunable linear correction, for $p \approx 1$.
In this sense, frustration-free Hamiltonians in such a ``percolated'' product-model on rectangular lattices resemble frustration-free natural Hamiltonians in the expected case as $p \to* 1$.

\section{Almost frustration-free Hamiltonians}
\label{sec:almostUnfrustrated}

The method of efficiently simulating ground state manifolds of frustration-free Hamiltonians can be 
extended to serve as a method to simulate almost-frustration-free Hamiltonians, albeit in a non-certified way.
Consider a Hamiltonian 
\begin{equation}
	H=H_0 + \lambda H_1
\end{equation}
for $\lambda\in \rr$ playing the role of a small perturbation, where
\begin{equation}
		 		H_0 \;=\! \sum_{\ens{a,b} \subset V}\! h_{a,b}
\end{equation}
is exactly frustration-free (i{.}e{.} in the sense defined in \sectionref{sec:prelimsFrustrationFree}), and $H_1$ is a small local perturbation. 
Then, one can still efficiently compute
\begin{equation}
	\inf_{\ket{\Phi}\in M} \bra{\Phi} H  \ket{\Phi} ,
\end{equation}
where $M$ denotes the (in general, degenerate) ground state manifold of $H_0$.
Again, we may characterize $M$ as the image of the low-dimensional subspace $\Symm(\cH\sox{n_c})$ under a tree-tensor network, as described in \sectionref{sec:sampling}; $H_1$ being a local Hamiltonian, each term of the infumum above can be
efficiently computed using a suitable basis of $\Symm(\cH\sox{n_c})$.
This is a variational approach that will always provide an upper bound to the 
true ground state energy.

In this way, one approximates the ground state manifold of an almost frustration-free
Hamiltonian with the ground state manifold of an
exactly frustration-free one. The interesting aspect here is that
one can consider the {\it image of an entire large subspace under a tensor network}.
In practice, one would think of a Hamiltonian $H_U$ near to a realistic one $H$, where one may show that $H_U$ is frustration-free (which may be efficiently verified using the algorithm of Ref.~\cite{Bravyi06}, as outlined in \autoref{sec:reviewBravyi}), and then approximate the ground state of the full Hamiltonian.
This approach appears to be particularly suitable for slightly frustrated Hamiltonians
reminding of Shastry-Sutherland type \cite{SS}
models, with ---  in a cubic lattice and a frustration-free Hamiltonian --- an additional  
bond along the main diagonal renders
the model frustrated.

\section{Summary}

In this work, we have investigated in great detail a class of models whose ground-state manifolds can be completely identified: those of physically realistic frustration-free models of spin-$1/2$ particles on a general lattice.
We have seen that the entire ground state manifold can be parametrized by means of tensor networks applied to symmetric subspaces, by essentially undoing a sequence of isometric reductions.
We also found that any ground state of such a system satisfies an area law, and hence contains little entanglement.
This is a physically meaningful class of physical models --- beyond the case of free models --- for  which such an area law behaviour can be rigorously proven.
It is the hope that the idea of considering entire subspaces under tensor networks, and eventually looking at the performance when being viewed as a numerical method, will give rise to new insights into almost frustration-free models.

\section{Acknowledgements}

This work has been supported by the EU (QESSENCE, MINOS, COMPAS) and the EURYI award scheme. We would like to 
thank D.\ Gross and S. Michalakis for discussions. Part of this work was done while JE was visiting the KITP in Santa Barbara
as participant of the Quantum Information Program. 

\bibliographystyle{abbrv}

\begin{thebibliography}{alpha}


\bibitem{Review}
	J.\ Eisert, M.\ Cramer, and M.~B.\ Plenio, Rev.\ Mod.\ Phys. {\bf 82}, 277 (2010).
			 	
\bibitem{Scholl}
	U.\ Schollwoeck,
	Rev.\ Mod.\ Phys.\ {\bf 77}, 259 (2005).
	
\bibitem{2d}
	F.\ Verstraete, J.\ I.\ Cirac, and V.\ Murg,
	Adv.\ Phys.\ {\bf 57},143 (2008).	
	
\bibitem{Wilczek}
	C.\ Holzhey, F.\ Larsen and F.\ Wilczek, Nucl.\ Phys.\ B {\bf 424}, 443 (1994).

\bibitem{Bombelli}		
	L.\ Bombelli, R.~K.\ Koul, J.\ Lee, and R.\ Sorkin,
	Phys.\ Rev.\ D
	{\bf 34}, 373 (1986).

\bibitem{Srednicki}		
	M.\ Srednicki, 
	Phys.\ Rev.\ Lett.\ {\bf 71}, 666 (1993).
	
\bibitem{Harmonic}		
	K.~M.~R.\ Audenaert, J.\ Eisert, M.~B.\ Plenio, and R.~F.\ Werner,
	Phys.\
	Rev.\ A {\bf 66}, 042327 (2002).

\bibitem{Latorre}		
	G.\ Vidal, J.~I.\ Latorre, E.\ Rico, and A.\ Kitaev, 
	Phys.\ Rev.\
	Lett.\ {\bf 90}, 227902 (2003).

\bibitem{Jin}		
	B.-Q.\ Jin and V.\ Korepin,  J.\ Stat.\ Phys.\ {\bf 116}, 79 (2004).

\bibitem{Area1}
	M.~B.\ Plenio, J.\ Eisert, J.\ Dreissig, and M.\ Cramer,	
	Phys.\ Rev.\ Lett.\ {\bf 94}, 060503 (2005).
	
\bibitem{Area2}
	M.\ Cramer and J.\ Eisert, New J.\ Phys.\ {\bf 8}, 71 (2006).
		
\bibitem{Cardy}		 
	P.\ Calabrese and J.\ Cardy, 
	J.\ Stat.\ Mech.\ P06002 (2004).
	
\bibitem{Fermi1}
	M.~M.\ Wolf, Phys.\ Rev.\ Lett.\ {\bf 96}, 010404 (2006).
		
\bibitem{Fermi2}
	M.\ Cramer, J.\ Eisert, and M.~B.\ Plenio,
	Phys.\ Rev.\ Lett.\ {\bf 98}, 220603 (2007).
			
\bibitem{Fermi3}
	D.\ Gioev and I.\ Klich, Phys.\ Rev.\ Lett.\ {\bf 96}, 100503 (2006).	

\bibitem{Hastings}
	M.~B.\ Hastings, JSTAT P08024 (2007).
	
\bibitem{Masanes}
	Ll.\ Masanes, Phys.\ Rev.\ A {\bf 80}, 052104 (2009).	

\bibitem{NPH}
	More precisely, one can construct problems of this sort which are \QMA-complete.
	\QMA\ is the class of problems which one obtains if one generalizes the notorious complexity 
	class \NP, to also allow algorithms which act on quantum states and which have a bounded probability of failure~\cite{Watrous}.
	An efficient and deterministic algorithm for sampling the ground states for the models presented in 
	Ref.~\cite{Power}  would thus imply $\P = \NP$.
	
\bibitem{Watrous}
	J.\ Watrous,  arXiv:0804.3401. 
	
\bibitem{Power}		
	D.\ Aharonov, D.\ Gottesman, S.\ Irani, and J.\ Kempe,
	Comm.\ Math.\ Phys.\ {\bf 287}, 41 (2009).
	
\bibitem{Shor}
	R.\ Movassagh, E.\ Farhi, J.\ Goldstone, D.\ Nagaj, T.~J.\ Osborne, and P.~W.\ Shor,
	Phys.\ Rev.\ A {\bf 82}, 012318 (2010).
		
\bibitem{Bravyi06}	
	S.~Bravyi, quant-ph/0602108.

\bibitem{OurPRL}
	N.\ de Beaudrap, M.\ Ohliger, T.\ J.\ Osborne, and J.\ Eisert,
	Phys.\ Rev.\ Lett.\ {\bf 105}, 060504 (2010).
		
\bibitem{NoGoMBQC}
	J.~Chen, X.~Chen, R.~Duan, Z.~Ji, B.~Zeng,
	arXiv:1004.3787.

\bibitem{Giov}
	P.\ Silvi, V.\ Giovannetti, S.\ Montangero, M.\ Rizzi, J.~I.\ Cirac, and R.\ Fazio,
	Phys.\ Rev.\ A {\bf 81}, 062335 (2010).		

\bibitem{APT79}		
	B.~Aspvall, M.~F.~Plass, and R.~E.~Tarjan,
	Inf.\ Proc.\ Lett.\ \textbf{8}, 121 (1979). 

\bibitem{Karp72} 		
	R.\ M.\ Karp, {\it Complexity of Computer Computations},  pp.\ 85--103, 
	in R.\ E.\ Miller and J.\
	W.\ Thatcher, Eds., {\it The IBM Research Symposia Series} (Plenum Press, 
	New York, 1972).
		
\bibitem{Footnote}
	Except where explicitly noted, references to the ranks or kernels of $2$-local operators $h_{a,b}$ 
	(acting on spins $a$ and $b$) are to be understood as applying to these operators as they act on 
	$\C^2 \otimes \C^2$ (i.e.,~on the spins $a$ and $b$ alone).
		
\bibitem{kQSat}
	A.\ Ambainis, J.\ Kempe, and O.\ Sattath, 
	arXiv:0911.1696.
	
\bibitem{RandomQSat}	
	C.~R.\ Laumann, A.~M.\ L{\"a}uchli, R.\ Moessner, A.\ Scardicchio, and S.~L.\ Sondhi,
	arXiv:0910.2058, Phys.\ Rev.\ A {\bf 81}, 062345 (2010).
			
\bibitem{BravyiQSat}			
	S.\ Bravyi, C.\ Moore, and A.\ Russell,
	arXiv:0907.1297.

\bibitem{Zittartz}
	A.\ Kl{\"u}mper, A.\ Schadschneider, and J.\ Zittartz,
	Z.\ Phys.\ B {\bf 87}, 281 (1992).
				
\bibitem{MERA}
	G.\ Vidal, Phys.\ Rev.\ Lett.\ {\bf 99}, 220405 (2007).

\bibitem{Footnote2}
	We effectively ignore the presence of single-spin terms in much of our analysis. However, such terms 
	impose constraints on the ground-state manifold. In particular, they imply that the Hamiltonian $H$ is frustration-free only if, for any ground state vector 
	$\ket{\Phi}$ of $H$, the spins on which such terms act are disentangled from the rest of the system.
	
\bibitem{Footnote3}
	In the case that $H$ contains single-spin terms, the resulting 
	Hamiltonian may still be frustrated, e{.}g{.} if there are no product states which are also in the kernel of all of the single-spin terms.
	Extending the analysis of Ref.~\cite{Bravyi06}, this may be efficiently determined 
	with no additional effort, by incorporating the single-spin constraints 
	when constructing a product state in the kernel of the Hamiltonian terms, and verifying that these constraints may be simultaneously satisfied.
	
\bibitem{Factorization}			
	R.\ Rossignoli, N.\ Canosa, and J.~M.\ Matera,
	Phys.\ Rev.\ A {\bf 80}, 062325 (2009).


\bibitem{Tarrach}
	A.\ Sanpera, R.\ Tarrach, and G.\ Vidal, quant-ph/9707041.

\bibitem{SM}
	J.\ Eisert and H.\ J.\ Briegel, Phys.\ Rev.\ A {\bf 64},  022306 (2001).	

\bibitem{Grimmett}
	G.\ Grimmett, \emph{Percolation} (Springer, New York, 1999).
	
\bibitem{Kieling}
	K.\ Kieling and J.\ Eisert, 
	{\it Percolation in quantum computation and communication}, in 
	{\it Quantum percolation and breakdown}, 
	Lecture Notes in Physics (Springer, Heidelberg, 2008),
	arXiv:0712.1836.

\bibitem{Footnote4}
	By the analysis of Ref.~\cite{RandomQSat}, one expects that such Hamiltonians 
	will be frustrated with probability $1$ if the ``natural'' terms form more than one 
	cycle in the lattice, which occurs with high probability if $p > 1\!/2$ (or more precisely, 
	if $p$ is greater than the percolation threshhold for the latice model).

\bibitem{SS}
	B.~S.\ Shastry and B.\ Sutherland, Physica
	{\bf 108}B, 1069 (1981).	
	
\end{thebibliography}

\appendix
\section{Technical lemmas}
\label{apx:techincalLemmas}

We now supply the proofs of technical lemmata required in the preceding sections.

	\begin{lemma*}[\autoref{lemma:nullOperatorProducts}]
		For two-spin state vectors $\ket{\psi}$ and $\ket{\phi}$, we have 
		$\big( \bra{\psi} \ox  \id \big) \big(  \id \ox \ket{\phi} \big) = 0$ only if both $\ket{\psi}$ and $\ket{\phi}$ are product states.
	\end{lemma*}
 
		\begin{proof}
			Consider Schmidt decompositions 
			\begin{equation}
				\ket{\psi} = \sum_r \mu_r \ket{e_r}\ket{f_r}, \ket{\phi} = \sum_s \nu_s \ket{g_s}\ket{h_s}, 
			\end{equation}
			where without loss of generality we may require $\bracket{f_r}{g_r} \ne 0$ by an appropriate choice of labels.
			Then we have
			\begin{equation}
				\bigl( \bra{\psi}  \ox  \id \bigr) \bigl(  \id \ox \ket{\phi} \bigr)
				=
					\sum_{r,s} \bigl( \mu_r \nu_s \bracket{f_r}{g_s} \!\bigr) \ket{e_r} \bra{h_s}	\;,
			\end{equation}
			which is only zero if $\mu_0 \nu_0 = \mu_1 \nu_1 = 0$, which implies $\mu_r = 0$ and $\nu_{1-r} = 0$ for some value of $r \in \ens{0,1}$.
		\end{proof}

	\begin{lemma*}[\autoref{lemma:isometricContractionPreservingProducts}]
		Let $U: \cH \to \cH \ox \cH$ be an isometry which is not a product operator.
		Let $\eta \ge 0$ be an operator on two spins, and \textup{$\eta' = \big(U\herm \ox \id_2) (\id_2 \ox \eta) (U \ox \id_2)$}.
		If $\eta'$ is not of full rank, then $\eta'$ is a product operator if and only if $\eta$ is a product operator.
	\end{lemma*}

		\begin{proof}
		Suppose that $\eta'$ is not full rank, and is a product operator.
		As it is positive semidefinite, $\eta'$ must either have the form 
		$\ket{\alpha}\bra{\alpha} \ox \eta''$ or $\eta'' \ox \ket{\alpha}\bra{\alpha}$ for some $\ket{\alpha} \in \cH$.
		In particular, there must exist a state vector $\ket{\gamma} \in \cH$ such that one of
		\begin{subequations}
		\begin{align}
			\bigl( \bra{\gamma} \ox  \id \bigr)  \eta'  \bigl( \ket{\gamma} \ox  \id \bigr) \;=&\; 0	\quad\text{or}
		\\
			\bigl(  \id \ox \bra{\gamma} \bigr)  \eta'  \bigl(  \id \ox \ket{\gamma} \bigr) \;=&\; 0		 
		\end{align}
		\end{subequations}
		holds.
		Decompose $\eta$ in its spectral decomposition,
		\begin{gather}
		 		\eta
			\;\;=\;\;
				\sum_k \lambda_k \ket{\phi_k}\bra{\phi_k}
		\end{gather}
		for $\lambda_k > 0$.
		Suppose that $\big( \bra{\gamma} \ox  \id \big) \eta' \big( \ket{\gamma} \ox  \id \big) = 0$: if we let $\ket{\Gamma} = U \ket{\gamma}$, we have 
		\begin{align}
				0 \;=&\;
	 			\bigl( \bra{\gamma} U\herm \ox  \id \bigr)  ( \id \ox \eta)  \bigl( U \ket {\gamma} \ox  \id \bigr)
			\notag\\=& \;
				\bigl( \bra{\Gamma} \ox  \id \bigr)
				\!\paren{ \sum_k \lambda_k  \id \ox \ket{\phi_k}\bra{\phi_k} }\!
				\bigl( \ket{\Gamma} \ox  \id \bigr)
			\notag\\=&\;\;
				\sum_k \lambda_k   T_k T_k\herm	\;,
		\end{align}
		where we define $T_k = \big( \bra{\Gamma} \ox  \id \big) \big(  \id \ox \ket{\phi_k} \big)$.
		By \autoref{lemma:nullOperatorProducts}, each operator $T_k$ is zero only if both $\ket{\Gamma}$ and $\ket{\phi_k}$ are 
		both product vectors; we may then 
		decompose $\ket{\Gamma} = \ket{\sigma}\ket{\tau}$ 
		and $\ket{\phi_k} = \ket{\tau'}\ket{\phi'_k}$, where we 
		require $\bracket{\tau}{\tau'} = 0$ for all $k$.
		We then have
		\begin{gather}
		 		\eta
			\; =\; 
				\ket{\smash{\tau'}}\bra{\smash{\tau'}} \ox \sqparen{ \sum_k \lambda_k   \ket{\smash{\phi'_k}}\bra{\smash{\phi'_k}} }.
		\end{gather}
		On the other hand, if $\big(  \id \ox \bra{\gamma} \big)  \eta  \big(  \id \ox \ket{\gamma} \big) = 0$, we obtain
		\begin{align}
			0  =&\;
	 			\bigl( U\herm \ox \bra{\gamma} \bigr)  \eta  \bigl( U \ox \ket {\gamma} \bigr)
			\notag\\[0.5ex]=& \;
				\sum_k \lambda_k   U\herm \bigl(  \id \ox \ket{\phi'_k} \bigr)\bigl(  \id \ox \bra{\phi'_k} \bigr) U \;,
		\end{align}
		for single-spin states $\ket{\phi'_k} = \big(  \id \ox \bra{\gamma} \big) \ket{\phi_k} $.
		We then require $U\herm \big(  \id \ox \ket{\phi'_k} \big) = 0$ for each $k$; because $U$ cannot be decomposed as $U' \ox \ket{u}$ for any state $\ket{u} \in \cH$, this implies that the vectors $\ket{\phi'_k}$ themselves are zero.
		Thus, $\ket{\phi_k} = \ket{\alpha_k}\ket{\beta}$ for some states $\ket{\alpha_k} \in \cH$ , and where $\bracket{\gamma}{\beta} = 0$.
		We then have
		\begin{gather}
		 		\eta
			\;\;=\;\;
				\sqparen{ \sum_k \lambda_k \ket{\smash{\alpha'_k}}\bra{\smash{\alpha'_k}}} \ox \ket{\beta}\bra{\beta} .
		\end{gather}
		In either case, $\eta'$ is a product operator only if $\eta$ is a product operator; the converse holds trivially.
 		\end{proof}


\end{document}